\documentclass[leqno, a4paper, 12pt]{article}
\usepackage[utf8]{inputenc}
\usepackage[margin=1in]{geometry}
\usepackage{enumitem,multirow}
\usepackage{amsmath,amssymb,amsfonts}
\usepackage{amsthm,subcaption}
\usepackage{algpseudocode}
\usepackage{pdfpages,comment}
\usepackage[flushleft]{threeparttable}
\usepackage{booktabs}
\usepackage{dcolumn}
\newcolumntype{d}[1]{D..{#1}}

\usepackage{xcolor}
\newtheorem{theorem}{Theorem}
\newtheorem{assumption}{Assumption}
\newtheorem{definition}{Definition}
\newtheorem{proposition}{Proposition}
\newtheorem{corollary}{Corollary}
\newtheorem{fact}{Fact}
\newtheorem{remark}{Remark}
\newtheorem{lemma}{Lemma}

\usepackage{bbm}
\usepackage{natbib}
\usepackage{algorithm}

\newcommand{\R}{\mathbb{R}}
\newcommand{\supp}{\operatorname{supp}}

\newcommand{\norm}[1]{\lVert #1\rVert}
\newcommand{\B}{\mathcal{B}}

\newcommand{\Hel}{\mathrm{H}}

\newcommand\independent{\protect\mathpalette{\protect\independenT}{\perp}}
\def\independenT#1#2{\mathrel{\rlap{$#1#2$}\mkern2mu{#1#2}}}

\usepackage{setspace}
\usepackage{graphicx}
\usepackage[colorlinks=true,linkcolor=blue,citecolor=blue,urlcolor=blue]{hyperref} 

\title{When is $p$-hacking detectable?}
\author{ Stefan Faridani\footnote{Georgia Institute of Technology. sfaridani6@gatech.edu. I thank Graham Elliott, Kaspar W\"uthrich, and the participants of the Georgia Tech Research Brownbag, BITSS annual meeting, and Georgia Econometrics Workshop for helpful comments. I also thank Xianfang Xiong for excellent research assistance. AI tools were consulted, but the paper was written by the author who takes full responsibility for all content and any errors.}}

\begin{document}

\maketitle 
\onehalfspacing
\begin{abstract}

We show that some forms of $p$-hacking and publication bias cannot be detected by examining the histogram of $t$-statistics or their $p$-values. Even when $p$-hacking is detectable, standard tests may lack power against some alternatives. We propose a novel test that is consistent against {\it every} form of selective reporting that is detectable from the distribution of reported $t$-statistics. Our test statistic is the distance between the smoothed empirical $t$-curve and the set of possible honest distributions.  Unlike existing tests which target roughness or incorrect derivatives in the $t$-curve and $p$-curve, the projection test is much less sensitive to rounding and derounding of the printed $t$-scores and we provide a novel deterministic derounding procedure that preserves its validity. In an empirical application to economics articles we show that the projection test rejects more often than any of  seven tests from \cite{Elliott} and \cite{Simonsohn,Simonsohn2015_stouffer}. Moreover, the projection test adds new knowledge by certifying that detection was not driven by rounding or derounding. In a simulation exercise that resamples actual medical trials we show that there are simple $p$-hacking techniques that the projection test detects but the others do not.

\end{abstract}

\vspace{0.1in}

\vspace{0.1in}

\pagenumbering{arabic}

\doublespacing

\clearpage
\section{Introduction}

Readers of an empirical literature want to know whether the results are real or exaggerated by selective reporting. Recent work examines histograms of reported $t$-statistics (the $t$-curve) or their $p$-values for telltale signs of selection such as bunching, jumps, wrong derivatives,  or roughness.\footnote{See for example: \cite{Head,Brodeur,havranek_housing,bb,Havranek24,Elliott,elliott2024powertestsdetectingphacking,Kudrinjmp,kudrinjmp24}}  But a more fundamental question remains: when is the $t$-curve enough to reveal selective reporting at all? 

Our answer has two parts. First, we show that simple forms of $p$-hacking and publication bias can produce a $t$-curve identical to one from an honest literature with much larger true effects. Since the consequences of selective reporting cannot be estimated from the $t$-curve without strong assumptions, we instead ask whether selective reporting left behind a trace. Such traces can be very faint in small samples (local power is usually trivial in all directions for all tests) so the right goal is to construct a test that can find every trace given enough data. We propose a test that achieves this frontier and call it the {\it projection test}. It is consistent against all detectable forms of $p$-hacking and publication bias. Applied to top economics journals it finds distortions across empirical methodologies that cannot be explained by sampling error or by rounding and derounding. Applied to a database of large medical trials, it finds no significant distortions.

Our first contribution is to show that simple forms of selective reporting are observationally equivalent to honest reporting of larger effects. Earlier work provides examples of selection where the {\it significant} part of the histogram of $p$-values is left clean \citep[Supplement~3]{Simonsohn}. In our examples the {\it entire} $t$-curve and $p$-curve have an honest explanation so there is in principle no way to detect the selection from them. For example, suppose that an editor publishes a result with probability equal to the posterior probability that the true estimand exceeded some threshold. We show that this produces a clean $t$-curve identical to one coming from a literature where ${\it all}$ true estimands exceed the threshold---even though many published effects in fact come from studies with small estimands. In this example the selection strategy remains undetectable regardless of the distribution of true effects, shifts the implied effect distribution by any amount along its support, and examining standard errors separately does not help even when they are assumed independent of the true coefficients. We provide several more simple examples, including one where reporting the larger of two effects is undetectable. \cite{kudrinjmp24} shows how to lower-bound the fraction of published results affected by selective reporting, while we show that no method based on the $t$-curve can upper-bound how much selective reporting has exaggerated effects without specific assumptions on its form.

The $t$-curve can reveal whether selective reporting left a trace. Many tests for such traces have been proposed. However, \cite{elliott2024powertestsdetectingphacking} and \cite{schimmack202} show that low power is a common problem. This paper's next contribution is to provide a theoretical explanation for this well-known phenomenon. We show that if true effects are spread over a range of values, a small amount of selection can always be closely mimicked by a small change in the distribution of true effects. This means that when the fraction of $p$-hackers is of order $1/\sqrt{n}$, no locally valid test has power greater than size in any direction. In the terminology of \cite{chensantos}, the honest model is locally just identified whenever the support of the true effects contains an interval. 

Since local power is usually trivial, an appropriate standard by which to compare tests is what kinds of selective reporting they can detect when the sample is large. We propose a novel test that achieves this frontier. The {\it projection test} is consistent against every detectable form of selective reporting.  Its test statistic is the distance between the smoothed empirical $t$-curve and the set of all possible honest $t$-curves. We reject the null when the empirical distribution lies too far from this set to be attributed to sampling error. Critical values are derived from the bootstrap of \cite{FangSantos}. Previous tests check  necessary conditions for honest reporting, e.g. a smooth $t$-curve with derivatives within bounds or close to a regularized version of itself \citep{caliperOriginal,Elliott,kudrinjmp24}. In contrast, the projection test is sharp and consistent because it checks directly for a $t$-curve that has an honest explanation or ``alibi."

The projection test also has an important practical advantage: it is robust to rounding. Published statistics are rounded, which creates small bumps in the t-curve that can cause false rejections. How to deal with this has sparked a debate \citep{KranzPutz2022,bb_reply,Elliott}. Results are often sensitive to whether and how derounding was done. \cite{elliott2024powertestsdetectingphacking} report that ``De-rounding
drastically changes the results and empirical conclusions.'' We  prove a bound on how much rounding can move the projection test statistic, so we can check whether a rejection could have been caused by rounding. This correction does not require adding noise, which improves replicability and computation time. Because the projection test does not target jumps or derivatives, our rounding correction costs little power in simulations.

We study two empirical applications. First, we apply our test to  19,155 $t$-scores reported by 641 articles published in top economics journals collected by \cite{bb}. These correspond to main hypotheses of interest studied by four research designs: randomized trials, instrumental variables, difference-in-differences, and regression discontinuity designs. We consider four derounding methods and compare to seven tests from \cite{Simonsohn,Simonsohn2015_stouffer} and \cite{Elliott}. The projection test is the only test to detect distortions beyond sampling error in all four research designs and under all four derounding methods. The tests of \cite{Elliott} and \cite{Simonsohn,Simonsohn2015_stouffer} do not find $p$-hacking or left-skew after uniform derounding in this subsample. This corroborates the findings in Table 3 of \cite{elliott2024powertestsdetectingphacking}.

We conduct a second exercise with a dataset of 17,680 $t$-tests from 5,310 large medical trials from the Cochrane Database of Systematic Reviews collected by \cite{Schwab2020}. In contrast to economics, we find no evidence of detectable selective reporting with any test. This setting is ideal for constructing a realistic simulation experiment. To compare the size and power of the tests, we repeatedly redraw 1500 trials from the dataset. This simulation DGP is realistic because it uses only real $t$-scores and thereby inherits the distribution of true effects, rounding, and within-trial dependence across outcomes of actual medical trials. 

We simulate $p$-hacking in the following way. For each trial one of its outcomes is randomly chosen to be its primary outcome. Some trials report their primary outcome honestly, while others selectively report by switching outcomes in one of two different ways from \cite{elliott2024powertestsdetectingphacking}. We show that, among the tests we compare, only the projection test has meaningful power while controlling size. The projection test can find hard-to-detect $p$-hacking because it is consistent.

The broad contribution of this paper is to trace out and try to achieve the frontier of what the $t$-curve can reveal. Some kinds of $p$-hacking and publication bias produce $t$-curves that appear to have much more evidential value than they actually do. Since upper-bounding its harm requires strong assumptions on its form, we focus instead on detecting the traces left by selective reporting. Local power is usually trivial, so even detectable $p$-hacking needs large sample sizes to be found. Therefore, an appropriate standard to measure tests by is which kinds of $p$-hacking and publication bias they are capable of detecting. We propose a new test that is consistent against every detectable alternative and robust to (de)rounding. Simulations show that it can find $p$-hacking that seven other tests cannot. Meta-analysts should add our projection test to their roster and become more cautious about interpreting a smooth $t$-curve as evidence of honesty.

\subsection{Related Literature}

This paper contributes to the meta-analytic literature on identifying and detecting publication bias and $p$-hacking and draws on the rich literatures on nonparametric deconvolution and identification.

We provide several cases of undetectable $p$-hacking and publication bias. To our knowledge, no prior work proves that a specific $p$-hacking process, or a publication rule based only on reported results, leaves the entire $t$-curve (and $p$-curve) observationally equivalent to fully honest reporting. Here we trace out the closest cases. \citet[Supplement~3]{Simonsohn} show that selecting for significance can leave the significant part of the $p$-curve honest-looking but Section \ref{sec:thresh_hacking} shows that this leaves a jump in the $t$-curve at the significance threshold and can therefore be detected by examining the whole $t$-curve. \cite{ulrich2015phacking} show that selecting the most significant test can fool some tests for evidential value but leave open the possibility that ``more fine-grained analyses of p-curves could be used to establish the presence of real effects." We can shed further light on this topic. The projection test detects such $p$-hacking in our simulations. But, Section \ref{sec:max_hacking} and Appendix \ref{sec:extended_max} show that this maximization-type $p$-hacking makes small effects look much larger than they are and is invisible to all tests when the distribution of true effects is sufficiently smooth. \cite{elliott2024powertestsdetectingphacking} and \cite{bruns2016pcurve} document cases where existing tests have little or no power but leave open whether a different test could perform better. We show that in some cases no test will. \cite{copas2004} state that the exact form of publication bias cannot usually be recovered without strong assumptions but this is not the same as finding publication bias whose presence is undetectable in principle as our Section \ref{sec:bayesian_pb} does. \cite{Andrews} provide sufficient conditions for identifying publication bias but do not propose specific selection mechanisms that are observationally equivalent to honest reporting. Section \ref{sec:bayesian_pb} demonstrates that the conditions of their Proposition 3 and their Appendix Section D.2 are substantive.

We show that no locally valid test for selective reporting of $t$-tests can have local power greater than size. Our proof parallels those of \cite{vanderVaart1998} and \cite{chensantos} and we owe our techniques to them. While the result appears standard from the perspective of nonparametrics, it is nevertheless important news for meta-analysis because it helps explain the widely reproduced phenomenon of low power to detect selective reporting when true effects are heterogeneous \citep{elliott2024powertestsdetectingphacking,schimmack202}. 

The projection test is distinct from previous tests for selective reporting because it is sharp. The tests of \cite{Simonsohn,Simonsohn2015_stouffer,caliperOriginal,Elliott,kudrinjmp24} and others all check the $p$-curve or $t$-curve for specific kinds of distortions including jumps, wrong derivatives, skew, or deviations from a smoothed version of itself. Their statistics target necessary conditions for honesty or evidential value. But the projection test directly tests whether the $t$-curve could have been produced honestly. This yields its consistency against all detectable alternatives, a property not established for existing tests. 


The projection test combines techniques from several distinct literatures. \citet[first version 2024]{kudrinjmp24} links the $p$-hacking detection problem to deconvolution by testing whether the empirical curve is far from a smoothed version of itself or a version constrained to satisfy shape restrictions. The projection test forges this link in a different way by measuring the distance between the empirical distribution and the closest honest one. We deconvolve using the discrete singular value decomposition of \cite{CarrascoPaper}.\footnote{See in addition: \cite{Fan91,on_Gibbs,Johannes,Hohage,faridani2025testingunderpoweredliteratures}} For critical values we show that the bootstrap of \cite{FangSantos} is valid in this context.

\section{Detectability}\label{sec:setup}

First define several pieces of notation. Let $\varphi(z)$ denote the probability density function (PDF) of the standard normal distribution. Adding a subscript $\varphi_{\sigma^2}(z)$ denotes the density of the normal with variance $\sigma^2$. No subscript means that the variance is unity.  Define $\mathcal{P}$ as the set of all probability distributions for a random variable in $\mathbb{R}$. Define $\mathcal{D}\subset \mathcal{P}$ as the set of all discrete distributions.

\subsection{Meta-Analyst's Null Hypothesis}

Our first goal is to formally state our null hypothesis. This requires characterizing what the $t$-curve looks like if there is no selective reporting. Consider a population of $t$-tests. Each researcher first draws the unobserved latent true effect $H\in\mathbb{R}$ from an unknown distribution $\Pi_0 \in \mathcal{P}$.\footnote{Here our $H$ corresponds to the signal-to-noise ratio SNR of \cite{vanZwet2021}.} We call $\Pi_0$ the {\it distribution of true effects}. The honest researcher then reports the $t$-statistic $T$, which is $H$ plus independent noise  $Z$.
\begin{equation}\label{eq:clean_tscore}
    T = H+Z  \qquad \text{where } H\independent Z\text{ and } Z\sim N(0,1)
\end{equation}

The probability density function  of $T$ is therefore the {\it convolution} of the distribution of $H$ with the standard normal density:
\begin{equation}\label{eq:clearn_tcurve}
    f_{T}(t)  =\mathbb{E}_{\Pi_0}\left[\varphi(t-H)\right]= \int_{-\infty}^\infty \varphi(t-h)  d\Pi_0(h)
\end{equation}

\begin{remark}\normalfont
   The exactly normal model in Equations (\ref{eq:clean_tscore}) and (\ref{eq:clearn_tcurve}) is standard. \cite{Elliott} and \cite{kudrinjmp24} use mixture setups with known distributions.
\end{remark}

Selective reporting means that the meta-analyst does not observe a random sample of $T$ because some $t$-statistics are not reported. Let $R$ be the event that $T$ is reported.  The meta-analyst instead observes draws from the conditional distribution of the reported $t$-statistics: $T \mid R$. Reporting is called {\it selective} if $R$ depends on the realization of $T$. For example, small values of $T$ may go unreported (e.g. by publication bias) or the researcher may repeatedly draw correlated $T$ with the same $H$ until they draw a large one and then stop. 

The form of selective reporting can be almost anything. The meta-analyst does not know what form it takes or the distribution of true effects and will not attempt to learn these. Instead, they wish only to determine whether the distribution of reported $t$-statistics can be explained without selective reporting. 

The meta-analyst's null hypothesis ${\mathbf{H}_0}$ states that there is some distribution of true effects $\Pi$ that explains the population $t$-curve without any selective reporting. This is the same null as \cite{kudrinjmp24} and suggested by \cite[pp. 2788–89]{Andrews}. We write this formally in Line (\ref{eq:h0}) below.
\begin{equation}\label{eq:h0}
     {\mathbf{H}_0:}\quad \exists \Pi \in \mathcal{P} \quad \text{ s.t. } \quad T \mid R  \stackrel{d}{=}  H'+Z\qquad H'\sim \Pi,\: Z\sim N(0,1) \qquad H'\independent Z 
\end{equation}

\noindent $\mathbf{H}_0$ says that $f_{T\mid R}$ has an honest explanation. Smoothness is necessary but not sufficient for this.\footnote{To be precise, ${\mathbf{H}_0}$ is true when the deconvolved function $e^{\zeta^2/2}\phi_{T\mid R}(\zeta)$ is a valid characteristic function where $\phi_{T\mid R}$ is the characteristic function of the reported $t$-statistic. This can only occur when $\phi_{T\mid R}$ has thin tails, i.e. the PDF $f_{T\mid R}$ is smooth.} Therefore if $f_{T\mid R}$ has a discontinuity or oscillates rapidly, then it could not have arisen ``naturally" and ${\mathbf{H}_0}$ is violated. While some violations of ${\mathbf{H}_0}$ are apparent to the naked eye (e.g. jumps or kinks), others are not \citep{simonsohn2020phacking}.

When ${\mathbf{H}_0}$ is false, then something has distorted the $t$-curve. However, the converse does not hold. If ${\mathbf{H}_0}$ is true, then either there is no selective reporting, or selective reporting is {\it undetectable}. We formally define the term  ``undetectable" below. 

\begin{definition}
    We say that selective reporting is {\bf undetectable} if $T \stackrel{d}{\neq }T \mid R$ but ${\mathbf{H}_0}$ is true.  We say that selective reporting is {\bf detectable} if ${\mathbf{H}_0}$ is false 
\end{definition}

Undetectable forms of selective reporting change the distribution of $T$, but keep it smooth and honest-looking. Thus, undetectable $p$-hacking yields a $t$-curve that could have been produced ``naturally" without any selective reporting by some other distribution of true effects. We stress that it is impossible to design a valid test based on the $t$-curve or $p$-curve that has power greater than size against undetectable forms of selective reporting. 

To build intuition, we will discuss several simple examples of $p$-hacking and publication bias  below. These cases were chosen because they illustrate some important lessons: 
\begin{enumerate}
    \item Some kinds of $p$-hacking are  undetectable regardless of $\Pi_0$ and other kinds are only undetectable when the distribution of true effects is sufficiently smooth.
    \item Some kinds of $p$-hacking are always detectable.
    \item  For any distribution of true effects with unbounded support, there is undetectable publication bias that shifts its center by as much as we like. This means that estimating the consequences of selective reporting using the $t$-curve is not possible without strong assumptions.\footnote{In contrast, lower-bounding the fraction of published results coming from researchers who selectively report is possible and done by \cite{kudrinjmp24}.} Detection is therefore at the frontier of what the $t$-curve can offer on its own.
\end{enumerate}

\subsection{Maximization $p$-hacking (Undetectable for smooth $\Pi_0$)}\label{sec:max_hacking}

Some kinds of $p$-hacking are undetectable when the distribution of true effects is smooth. Consider a $p$-hacking process where a researcher first draws a single unobserved true effect $H\sim \Pi_0$. Then they observe two jointly normal $t$-statistics $\{T_1,T_2\}$ both with marginal distributions $N(H,1)$. The researcher reports the larger of the two $t$-statistics (where ``larger" means rightmost on the number line). This simple type of $p$-hacking distorts the rejection probability of the underlying hypothesis test, but produces no jumps in the $t$-curve. It is already known that many existing tests can have low power to detect it for some distributions of true effects \citep{Elliott,elliott2024powertestsdetectingphacking, kudrinjmp24}.

We will now show that whether this $p$-hacking process is detectable at all depends on the distribution of $H$. Consider $H\sim N(0,1)$. Then, the distribution of the reported $t$-statistic is the maximum of two (joint) normals plus a third independent normal:
\begin{equation}\label{eq:maxhack1}
   T \mid R \stackrel{d}{=}  \underbrace{Z_1}_{\text{true effect}}+\underbrace{\max\{Z_2,Z_3\}}_{\text{noise (max. hacking)}}
\end{equation}

But this $t$-curve can be explained without any selective reporting! If instead $H$ is distributed as the maximum of two (joint) normals and there is no $p$-hacking, we get exactly the same $t$-curve:
\begin{equation}\label{eq:maxhack2}
  T \mid R \stackrel{d}{=}   \underbrace{Z_1}_{\text{noise (no hacking)}}+\underbrace{\max\{Z_2,Z_3\}}_{\text{true effect}}
\end{equation}
\noindent Notice that the only difference between (\ref{eq:maxhack1}) and (\ref{eq:maxhack2}) is which component is called the true effect and which is called noise. This proves that maximization $p$-hacking can be undetectable because it can change the $t$-curve (and the size of the tests) but the resulting $t$-curve has an alibi---i.e. it could be explained without $p$-hacking. More generally, this logic suggests that many location-equivariant selection rules will be undetectable when $H$ is a mixture of normals with variance at least one.

This is not a knife-edge case. Appendix \ref{sec:extended_max} formally extends this example. Suppose that the researcher reports only the maximum of $k$ $t$-scores that share a true effect $H$ and have common correlation $\rho\in [0,1)$. Let the true effects satisfy $H = A+S$ where $S\sim N(0,s^2)$, $A\independent S$ and $A$ is allowed to have any distribution. Then, maximization $p$-hacking is undetectable if $s^2 \geq (1-\rho)(1-1/k) $. Thus, some ``Gaussian content" in $\Pi_0$ can mask significant maximization $p$-hacking.

 The implication for practice is that even very simple forms of $p$-hacking can be totally undetectable by any examination of the $t$-curve or $p$-curve.

   This example is the first to our knowledge to show that $p$-hacking can be fully undetectable in principle from any examination of the published $t$-scores. \citet[Supplement~3]{Simonsohn} show that when a researcher draws uncorrelated test statistics and stops at the first significant result, the significant part of the $p$-curve and $t$-curve is unchanged. But, that $p$-hacking would still be detectable based on the full distribution of $p$-values because it leaves jumps at the thresholds.  

\subsection{Threshold $p$-hacking (Always Detectable)}\label{sec:thresh_hacking}

Other kinds of $p$-hacking will be detectable regardless of $\Pi_0$. For example, suppose that the researcher $p$-hacks only if their initial result is insignificant. Specifically, the researcher draws two (potentially correlated) $t$-statistics $\{T_1,T_2\}$ with the same $H$. They report $T_1$ if $T_1>\text{cv}$ where cv is some critical value. Otherwise, they report $\max\{T_1,T_2\}$. We will refer to this as ``threshold $p$-hacking."

Threshold $p$-hacking produces a jump in the $t$-curve $f_{T\mid R}$ at the critical value and many tests could potentially detect it. In fact, the caliper test \cite{caliperOriginal,Kudrinjmp} is designed specifically to target such jumps. To see why threshold $p$-hacking is always detectable, consider the PDF of any ``honest" $t$-curve in Equation \ref{eq:clearn_tcurve}. Since the normal density $\varphi$ is everywhere continuously differentiable, so is the density $f_T$. Since any jumps in the $t$-curve are incompatible with {$\mathbf{H}_0$}, threshold $p$-hacking is always detectable.

\subsection{Standard Error Shopping (Undetectable for any $\Pi_0$)}

 Consider a $p$-hacking process where a researcher always honestly reports the coefficient (numerator of the $t$-score), but selectively reports only the smallest of several possible standard errors (denominator of the $t$-score) of which only one is actually valid. For instance, they could select the level of clustering or bandwidth of the HAC kernel that results in the smallest standard error which is a $p$-hacking technique discussed by \cite{elliott2024powertestsdetectingphacking}. They show that for one-sided $t$-tests, standard error shopping is detectable because it creates a discontinuity. We now show that for two-sided $t$-tests, under some independence conditions, this kind of $p$-hacking is undetectable regardless of $\Pi_0$ and can distort the $t$-curve by such a large amount that it can make a literature composed of entirely true nulls appear to be a literature with many effective interventions. The reason is that standard error shopping doesn't depend on thresholds at all: the researcher simply exaggerates the precision of their estimates and tries to make the confidence intervals appear smaller than they really are.

The details of the example are the following. The researcher always honestly reports the study's estimate $\hat{\mu}$. Assume that $\hat{\mu}$  has true expectation $b\in\mathbb{R}$ and true standard deviation $\sigma > 0$ which can vary across researchers. Assume that $\hat{\mu}|(b,\sigma) \sim N(b,\sigma^2)$. Let $s_1>0$ be an alternative (potentially anticonservative) standard error. Let $H=b/\sigma \sim \Pi_0$. So $\hat{\mu}/\sigma=H+Z$ where $Z \sim N(0,1)$ and $Z \independent (H,\sigma,b,s_1)$. Thus $\hat{\mu}/\sigma$ satisfies $\mathbf{H}_0$. But instead of honestly reporting $\hat{\mu}/\sigma$, the researcher reports $T=\hat{\mu}/\min\{\sigma,s_1\}$. Then if we let $L  = \sigma / \min\{\sigma,s_1\}$, we have $T = L(H+Z)$ with $\mathbb{P}\left[L \geq 1\right]=1$ and $Z \independent (H,L)$. So by Lemma~\ref{lem:absorb}(ii),  $T \stackrel{d}{=} H'+Z'$ where $H'\independent Z'$ and $Z'\sim N(0,1)$. So $T$ satisfies $\mathbf{H}_0$ and standard error shopping is undetectable.  

The distortion can be large. Suppose that the literature contains only true nulls, i.e. $H = 0$ a.s. If $s_1=\sigma/2$, true nulls are rejected about 33\% of the time instead of 5\%, yet the t-curve is exactly that of an honest literature with effects $H'\sim N(0,3)$. The meta-analyst therefore cannot tell the difference between a heterogeneous literature filled with small and large effects and a null literature with standard error shopping.

\subsection{ Bayesian Publication Bias (Undetectable for any $\Pi_0$)}\label{sec:bayesian_pb}

Simple publication bias can severely misrepresent the evidential value of the reported results while remaining fully undetectable. Consider the class of selective reporting rules where $R$ depends only on the realized $T$ and some other information $Y$ which is independent of the sampling error $Z$ (e.g. the true standard error). This is an extension of the baseline publication bias model of \cite{Andrews}. By Bayes Rule the density of reported $t$-statistics  is:
\begin{equation}\label{eq:pub_bias}
   f_{T|R}(t) = \frac{\mathbb{E}\left[\mathbb{P}\left(R\,|T,Y\right)\mid T=t\right]f_{T}(t) }{\mathbb{P}\left(R\right)} 
\end{equation}

Researchers and publishers may screen which results they report by observing $(T,Y)$, making some inference about the true effect $H$, and then reporting based on that inference. Proposition \ref{prop:self_fulfilling_selection} shows that when the publication probability takes the form of a posterior expectation of a function $g(H,Y)$ (e.g. the posterior probability that the true effect $H$ clears some fixed threshold), this produces an honest $t$-curve corresponding to a potentially dramatically different distribution of true effects $\Pi'$ than the true one $\Pi_0$.

\begin{proposition}\label{prop:self_fulfilling_selection}
Let $Y$ be study characteristics taking values in a measurable space $\mathcal{Y}$, and let $g:\mathbb{R}\times\mathcal{Y}\to[0,1]$ be measurable with $\mathbb{E}[g(H,Y)]>0$. If $T=H+Z$ with $Z\independent(H,Y)$ and $\mathbb{P}(R\mid T,Y)=\mathbb{E}[g(H,Y)\mid T,Y]$, then
\[
(T,Y)\mid R\;\stackrel{d}{=}\;(H'+Z',\,Y') \text{ s.t. } (H',Y')\sim\Gamma',\ Z'\stackrel{d}{=}Z, \text{ and } Z'\independent(H',Y'),
\]
where $\Gamma'(B)=\mathbb{E}\big[g(H,Y)\mathbf{1}\{(H,Y)\in B\}\big]/\mathbb{E}[g(H,Y)]$ for measurable $B\subseteq\mathbb{R}\times\mathcal{Y}$. In particular, $T\mid R\stackrel{d}{=}H'+Z'$ with $H'\sim\Pi'$ and $Z'\independent H'$, where $\Pi'(A)=\Gamma'(A\times\mathcal{Y})=\mathbb{E}\big[g(H,Y)\mathbf{1}\{H\in A\}\big]/\mathbb{E}[g(H,Y)]$. Proof: Section \ref{proof:prop:self_fulfilling_selection}
\end{proposition}

  The contribution of Proposition \ref{prop:self_fulfilling_selection} is that a well-motivated class of selection is totally indistinguishable from honest reporting under a {\it different} distribution of true effects. This complements the sufficient conditions for identification provided by \cite{Andrews}. Their optimal editor (their Appendix D.6) thresholds the posterior mean and so leaves a detectable jump; the editor in Section \ref{sec:bayesian_pb} randomizes and leaves none. No restrictions on $\Pi_0$ are needed and the argument goes through for any distribution of $Z$. The proof is mathematically straightforward and implicitly requires that the selector has a correctly specified prior on $\Pi_0$.

The main real-world implication is that no test based on the $t$-curve can distinguish whether a literature studies only important effects or only important-looking ones. The leading example is described in Corollary \ref{cor:cutoff}, where results are more likely to be reported if they are likely to come from large true effects.

\begin{corollary}\label{cor:cutoff}
Let $T=H+Z$ with $Z\independent H$. If $\mathbb{P}\left(R|T=t\right) = \mathbb{P}\left(H\geq c|T=t\right)$ where $c$ satisfies $\mathbb{P}(H\geq c)>0$, then $T\mid R \stackrel{d}{=} H'+Z'$ where $H'\stackrel{d}{=} H|H\geq c$.
\end{corollary}

Corollary \ref{cor:cutoff} shows that two literatures can share an apparently honest $t$-curve but have very different true effects. Suppose that $\Pi_0 =N(0,1)$ and each result is reported with the posterior probability that $H\geq 0$. This produces exactly the same $t$-curve as if every true effect were positive and honestly reported, even though one third of the reported results actually correspond to negative true effects. 

There are two key implications  of Corollary \ref{cor:cutoff}. First, it is impossible to upper-bound the severity or consequences of selective reporting from the $t$-curve alone.\footnote{\cite{copas2004} show that if the fraction of results that are published is known then bounds can be constructed. In this paper we allow the fraction of results that are published to be anything.} When $\Pi_0$ has unbounded support we can make the gap between the medians of $\Pi_0$ and $\Pi'$ as large as we like undetectably. While lower-bounding the fraction of published studies that were subject to selective reporting is done by \cite{kudrinjmp24}, upper-bounding the consequences is not possible.  Second, detection is the frontier of what the $t$-curve can do and it is informative. While failing to reject ${\mathbf{H}_0}$ does not rule out consequential selection, a rejection shows distortion beyond sampling error. 

Corollary \ref{cor:cutoff2} also shows that observing true standard errors does not help detectability for Bayesian publication bias when selection is based on the posterior probability that the true coefficient $b$ exceeds a threshold. Assuming that true coefficients and true standard errors are independent does not help. 

\begin{corollary}\label{cor:cutoff2}
Let $\hat\mu\mid(b,\sigma)\sim N(b,\sigma^2)$, $H=b/\sigma$, and $T=\hat\mu/\sigma=H+Z$, so that $Z\sim N(0,1)$ and $Z\independent(b,\sigma)$. If $\mathbb{P}(R\mid T,\sigma)=\mathbb{P}(b\geq c\mid T,\sigma)$ where $c$ satisfies $\mathbb{P}(b\geq c)>0$, then $T\mid R\stackrel{d}{=}H'+Z'$ where $H'\stackrel{d}{=}H\mid b\geq c$, $Z'\stackrel{d}{=}Z$, and $Z'\independent H'$. Moreover, $(\hat\mu,\sigma)\mid R\stackrel{d}{=}(b'+\sigma' Z',\,\sigma')$ where $(b',\sigma')\stackrel{d}{=}(b,\sigma)\mid b\geq c$ and $Z'\independent(b',\sigma')$. If $b\independent\sigma$, then $b'\independent\sigma'$ and $\sigma'\stackrel{d}{=}\sigma$.
\end{corollary}

    Corollary \ref{cor:cutoff2} does not contradict Proposition 3 of \cite{Andrews} because their Definition 1 requires that publication bias select studies only via their realization of $T$ and nothing else. It also does not contradict their Propositions 2 or D.1 because there are no replication experiments here. Our contribution here is to show that when the selector sees as much information as the meta-analyst, then a simple and plausible selection mechanism is undetectable even when true standard errors and true coefficients are assumed to be independent.



\subsection{Visual Intuition for Detectability}

 Figure \ref{fig:detectability} visualizes detectability in three examples. All three simulated $t$-curves were subject to severe $p$-hacking but the detectability of that hacking varies widely. The red solid line is the $t$-curve after threshold-type $p$-hacking where $H=0$ and the researcher reports the first significant result out of 8 tries and otherwise reports the last one. This is easy to detect because of the jump in the curve caused by the threshold. The solid blue line is the $t$-curve after maximization-type $p$-hacking where the researcher reports the largest of eight independent $t$-statistics when $H=0$. The distortion is smooth, and not immediately visually apparent. Nevertheless $p$-hacking is still detectable here because the variance of $T|R$ is $0.37$ which is less than the honest minimum of $1$. Finally, the dashed blue line is an identical DGP to the solid blue line except $H\sim N(0.5,1)$. Now the smoothness of the distribution of true effects fully masks the selective reporting and the $p$-hacking is undetectable despite its severity.
 
\begin{figure}
    \centering
    \includegraphics[width=0.7\linewidth]{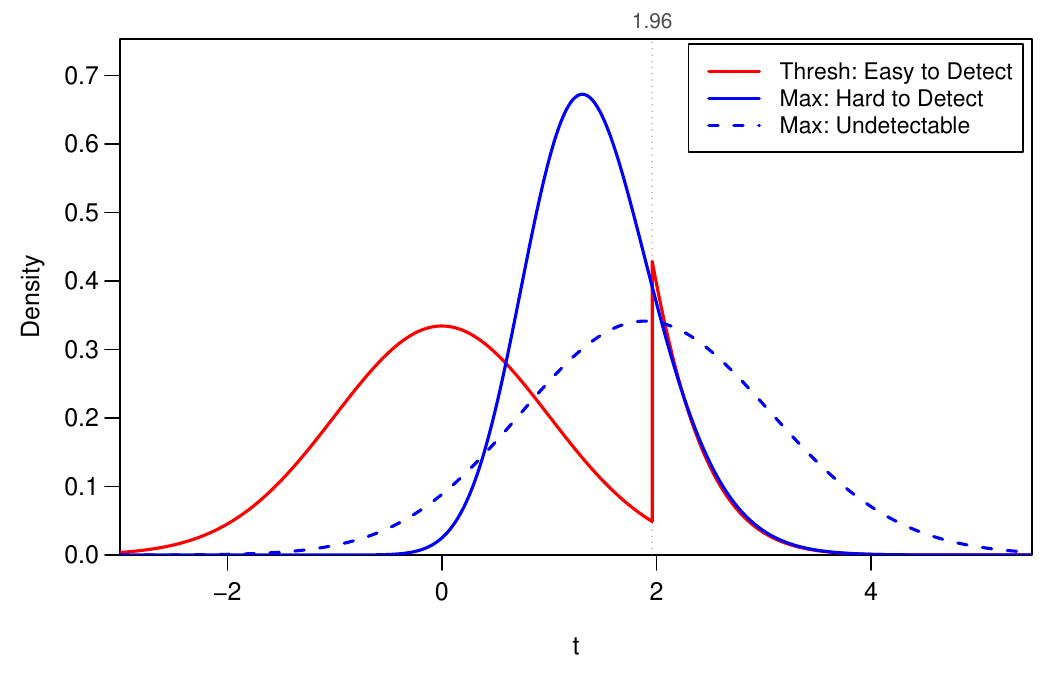}
    \caption{Three examples of severe $p$-hacking with varying detectability.}
    \label{fig:detectability}
\end{figure}

\subsection{Everything is Locally Undetectable}\label{sec:local}

We have seen that some specific forms of selective reporting leave a $t$-curve that looks honest. A weaker version of the same phenomenon applies in general whenever the support of the distribution of true effects contains an interval: small amounts of selection always have an alibi to first order. In other words, when the share of $p$-hackers is of the same order as $n^{-1/2}$, no locally valid test based on the $t$-curve, ours or any other, can detect it. 

Formally, our claim is that no locally valid test of $\mathbf{H}_0$ has local power greater than size in any direction when the support of $\Pi_0$ contains an interval. Appendix \ref{app:local_untestability} verifies this by paralleling key results from \cite{chensantos} and \cite{vanderVaart1998}.  To our knowledge, this connection has not been made in the literature on detecting selective reporting. There are several key implications for theory and practice:

\begin{enumerate}
    \item Low power is a fundamental property of the problem. \cite{elliott2024powertestsdetectingphacking} find that the power of existing tests depends heavily on the form of $p$-hacking and on the distribution of true effects.  Appendix \ref{app:local_untestability} confirms that in order for any test to have power it is necessary (but not sufficient) that the selective reporting is severe relative to sampling noise, i.e. the fraction of selective reporters is large relative to $n^{-1/2}$. So low power against mild selective reporting is not a consequence of choosing the wrong test. 

    \item Consistency against fixed alternatives is a natural test property to strive for. When the support of $\Pi_0$ contains an interval, every test that controls size locally has trivial local power, so tests cannot be ranked by local power.  Thus, we ask that a test be consistent against every fixed detectable alternative, which the projection test proposed later in this paper delivers. 

    \item In a heterogeneous literature, a non-rejection is weak evidence against selective reporting.  A rejection, on the other hand, signals distortion  large relative to sampling noise. 


\end{enumerate}

\section{The Projection Test}\label{sec:the_test}

In this section we propose the {\it Projection Test.} This test is ``sharp" in the sense that it exploits all available information in the $t$-curve and is consistent against every kind of detectable selective reporting. The intuition is simple: reject the null hypothesis if the observed $t$-curve is farther from the closest possible ``honest" $t$-curve than is likely to occur by chance. In other words, we conclude there is $p$-hacking if the $t$-curve lacks an ``alibi."

Formally, there are five main steps to developing the test. First, we define a Hilbert space that allows us to measure the distance between PDFs and to convert distributions to vectors. Second, we show that $\mathbf{H}_0$ is equivalent to testing whether $f_{T \mid R}$ belongs to a known convex set of PDFs. Third, we show that this is also equivalent to testing whether a vector of moments of the data belongs to another known convex set in Euclidean space. In other words, the moments contain all of the information about $\mathbf{H}_0$. Fourth, we compute critical values using the bootstrap of \cite{FangSantos} and show validity and power of an idealized test. Finally, we show that the feasible numerical implementation of the idealized test is also valid and consistent. 

\subsection{Preliminaries: The Hilbert Space $\mathcal{L}_T$}\label{sec:Hilbert}

In order to test $\mathbf{H}_0$ we will need a way to convert PDFs into vectors. This is possible if we draw on the theory of Hilbert Spaces. Our main sources are Section 2 of \cite{CARRASCO2Handbook} and Example 1 of \cite{CarrascoPaper}.

First we define a notion of ``distance" between two PDFs. Define the norm  $||\cdot||_T$ that maps a real-valued function $\phi$ to a real number (or $\infty$) and is defined as:\begin{equation}
    ||\phi||_T \equiv \sqrt{\int_{-\infty}^\infty \phi(t)^2\varphi(t)dt} 
\end{equation}
\noindent  We have chosen this norm to measure the distance between two PDFs because it can be converted to a sum of squared expectations later on.  Define the Hilbert Space $\mathcal{L}_T$ as the space of all functions with finite $||\cdot||_T$ and equip it with the inner product $\langle \cdot, \cdot\rangle_T$ that induces the norm $||\cdot||_T$.
\begin{align}
    \mathcal{L}_{T} &\equiv \left\{\phi(x)\text{ such that } ||\phi||_T < \infty \right\} \\
     \langle \phi_1, \phi_2 \rangle_{T} &\equiv  \int_{-\infty}^\infty \phi_1(x)\phi_2(x)\varphi(x)dx\\
     \langle \phi, \phi \rangle_T &= ||\phi||_T^2
\end{align}

\noindent $\mathcal{L}_T$ is a very large set of functions. It includes, for example, all PDFs that have bounded heights and indeed all PDFs in $\mathcal{L}_2$. 

Members of this space can be expressed in terms of basis functions. Let $\psi_j$ be the normalized Hermite polynomials $\psi_j(t) =\frac{1}{\sqrt{j!}} He_j\left(t\right)$.\footnote{The  Hermite Polynomials are defined as: $
He_j(t)=\sum_{l=0}^{[j/2]}(-1)^l \frac{(2l)!}{2^ll!}\binom{j}{2l}t^{j-2l}$. } Example 1 of \cite{CarrascoPaper} shows that $\psi_j$ are a complete orthonormal basis for $\mathcal{L}_T$:
\begin{align}
\langle \psi_j, \psi_k\rangle_T &= \mathbf{1}\left\{j=k\right\} \label{eq:normalpsi}
\end{align}
\noindent Since $\psi_j$ form a complete orthonormal basis, any member of $\mathcal{L}_T$ can be expressed in terms of a (possibly infinite)  linear combination of complete orthonormal basis functions:
\begin{align}
    \phi(t) &= \sum_{j=0}^\infty \langle \phi,\psi_j \rangle_T\:\psi_j(t)\quad \forall \phi \in \mathcal{L}_T \label{eq:polynomial_representation}
\end{align}

The basis polynomials are useful because they allow us to express the norm $||\cdot||_T$ as a sum of inner products over them. Theorem 2.8 of \cite{CARRASCO2Handbook} (Parseval's Formula) says:
\begin{align}
        ||\phi||_T^2 &=    \sum_{j=0}^\infty \langle \phi,\psi_j \rangle_T^2 \quad \forall \phi \in \mathcal{L}_T  \label{eq:norm_as_sum}
\end{align}

Inner products are helpful because they can be interpreted as expectations. Consider any random variable $A$ with PDF $f_A \in \mathcal{L}_T$. Since integrals over densities are expectations by definition, the inner products correspond to expectations:
\begin{align}
    \langle f_A,\psi_j \rangle_T = \int_{-\infty}^\infty f_A(t)\psi_j(t)\varphi(t)dt = \mathbb{E}\left[\psi_j(A)\varphi(A)\right]\label{eq:integral_as_expectation}
\end{align}

The purpose of setting up this Hilbert Space is to express the distance between two PDFs in terms of expectations. For example, if $A,B$ are random variables with PDFs in $\mathcal{L}_T$, then combining Equation (\ref{eq:norm_as_sum}) with (\ref{eq:integral_as_expectation}):
\begin{align}
    ||f_A-f_B||_T^2 &=    \sum_{j=0}^\infty (\mathbb{E}\left[\psi_j(A)\varphi(A)\right]-\mathbb{E}\left[\psi_j(B)\varphi(B)\right])^2
\end{align}

These expectations are easy to work with because Lemma \ref{lem:bound_coeffs} shows that they are taken over uniformly bounded random variables---regardless of the distribution of $A$. 

\begin{lemma}\label{lem:bound_coeffs}
$$\sup_{t\in\mathbb{R}}\left|\psi_j(t)\varphi(t)\right| \leq \frac{1}{\sqrt{2\pi}}\qquad \forall j\in \{0,1,\cdots\}$$  

    Proof: Section \ref{proof:lem:bound_coeffs}
\end{lemma}

\subsection{Sharp Testable Restriction}

We can use the properties of the Hilbert space $\mathcal{L}_T$ developed in the previous section to convert $\mathbf{H}_0$ into a testable restriction on a vector of moments of $T \mid R$. First we need to impose one mild assumption on the $p$-hacking process: that it results in a reported $t$-statistic with a PDF that does not explode too quickly to infinity:

\begin{assumption}\label{assum:cts_TR}
    $T \mid R$ is a continuous random variable with PDF $f_{T \mid R}\in \mathcal{L}_T$.
\end{assumption}

\noindent Assumption \ref{assum:cts_TR} is very weak for several reasons. These are discussed in the remarks below.

\begin{remark}\normalfont
    The first reason that  Assumption \ref{assum:cts_TR} is mild is  that Bayes' Rule reveals that it holds in every case of interest. If $\mathbb{P}[R]>0$, then $T \mid R$ is continuous with PDF   $f_{T \mid R}(t) = \frac{f_T(t)\mathbb{P}[R|T=t]}{\mathbb{P}[R]}$, so $f_{T \mid R}$ is in $\mathcal{L}_T$ if $f_T$ is. If $T=H+Z$ where $Z$ has a PDF $f_Z$ that is bounded above (e.g. normal, Student-$t$, etc), then $f_T \in \mathcal{L}_T$ and Assumption \ref{assum:cts_TR} is satisfied. 
\end{remark}

\begin{remark}\normalfont
     The second reason that  Assumption \ref{assum:cts_TR} is mild is that in practice it will hold mechanically whenever the meta-analyst de-rounds $t$-statistics by adding uniform noise as in \cite{bb}. 
\end{remark}

\begin{remark}\normalfont
Assumption \ref{assum:cts_TR} is guaranteed to hold under the null hypothesis. Its primary technical function is that it makes the power of the test more straightforward to study without adding a substantive restriction.
\end{remark}


 We now state the testable restriction in Equation (\ref{eq:testable_restriction}) below. Recall that $\mathcal{D}$ was defined as the set of all discrete distributions of real-valued random variables. The restriction says that the distance between  $f_{T\mid R}$ and the set of all possible honest densities $f_T$ for which $H$ has a discrete distribution $\Pi \in \mathcal{D}$ is zero. This essentially means that the observed $t$-curve cannot be separated from the possible set of undistorted $t$-curves. 
 \begin{equation}\label{eq:testable_restriction}
 \inf_{\Pi \in \mathcal{D} } \left|\left|f_{T\mid R}- \int_{-\infty}^\infty \varphi(t-h)d\Pi(h) \right|\right|_T =0
\end{equation}

\begin{remark}
    \normalfont Even though the minimization in  (\ref{eq:testable_restriction}) is over all discrete distributions, we do not assume that $\Pi_0$ is discrete. Intuitively, minimization over discrete distributions yields the same infimum as minimizing over all distributions because the discrete distributions are dense given this norm.
\end{remark}

Theorem \ref{thm:testable_implication} below shows that Restriction  (\ref{eq:testable_restriction}) is sharp, i.e. equivalent to the null hypothesis. This means that selective reporting is detectable if and only if it separates $f_{T\mid R}$ from the set of possible honest $t$-curves in the $||\cdot||_T$ norm.
\begin{theorem}\label{thm:testable_implication} If Assumption \ref{assum:cts_TR} holds, then:
 $$\mathbf{H}_0 \iff  (\ref{eq:testable_restriction})$$
  Proof: Section \ref{proof:thm:testable_implication}
\end{theorem}

Restriction (\ref{eq:testable_restriction})  is stated in terms of densities, which makes it easy to interpret but hard to see how it can be tested empirically. We can use the facts about Hilbert spaces from Section \ref{sec:Hilbert} to rewrite the $||\cdot||_T$ distance between PDFs as a difference between expectations. Define $\theta_0$ as the vector of expectations with elements below. Here the element $\theta_{0,j}$ is the $j^{\text{th}}$ weighted Hermite coefficient of the observed $t$-curve.
\begin{equation}
    \theta_{0,j}\equiv \mathbb{E}\left[\psi_j(T)\varphi(T)\mid R \right]
\end{equation}

For any ${\Pi}\in \mathcal{P}$ and any $j\in \{0,1,\cdots\}$, define $c_j(\Pi)$ as:
\begin{align}
    c_j(\Pi)  &\equiv \mathbb{E}_{\Pi}\left[\psi_j(H+Z)\varphi(H+Z) \right] \quad\text{where $H\sim \Pi$, $Z\sim N(0,1)$ and $H\independent Z$}
\end{align}

\begin{remark}\normalfont
    We use the Hermite polynomials $\psi_j$ as our basis functions because they are the singular functions of the Gaussian convolution \citep{CarrascoPaper}. This fact underpins Lemma \ref{lem:cj_Pi} which will allow us to map basis representations of the distribution of $T$ into representations of the distribution of $H$ one-to-one. 
\end{remark}

Lemma \ref{lem:ynorm_infsum} now expresses the distance $||\cdot||_T$ between two PDFs as the sum of squared differences in expectations. 

\begin{lemma}\label{lem:ynorm_infsum}

If Assumption \ref{assum:cts_TR} holds, then:

    $$  \left|\left|f_{T\mid R}- \int_{-\infty}^\infty \varphi(t-h)d\Pi(h) \right|\right|_T =  \sqrt{  \sum_{j=0}^\infty    \left(\theta_{0,j} -c_j(\Pi)\right)^2 } < \infty$$



    Proof: Section \ref{proof:lem:ynorm_infsum}
\end{lemma}
\noindent The utility of Lemma \ref{lem:ynorm_infsum} is simple: instead of comparing entire density functions, we can compare expectations of known, bounded functions of the data. Testing whether two densities are close reduces to testing whether two vectors of sample means are close.


In practice, the sample analogue of the infinite sum in Lemma \ref{lem:ynorm_infsum} behaves poorly because if we were to plug a vector of sample means $\widehat{\theta}_n$ in for $\theta_0$, the infinite sum may diverge. To avoid this issue, we weight each coordinate $j$ with weights $w_j \in (0,1]$ that decay with $j$. Specifically let the weights satisfy:

\begin{assumption}\label{assum:weights}
    Let $0< w_j \leq 1 $ for all $j \in \{0,1,\cdots\}$ and let $\sum_{j=0}^\infty w_j^2 < \infty$.
\end{assumption}

We will eventually show that any set of weights that satisfy Assumption \ref{assum:weights} will result in a consistent test. Now define $d_w(\theta)$ as the weighted $\mathcal{\ell}_2$ distance from any point $\theta$ to the set of points in the set allowed under the null hypothesis:
\begin{equation}\label{eq:Dj_definition}
    d_w(\theta) \equiv   \inf_{\Pi \in \mathcal{D}} \sqrt{ \sum_{j=0}^\infty w_j^2(\theta_j-c_j(\Pi))^2}
\end{equation} 

\begin{remark}\normalfont
    In practice, the researcher does not need to adjust or calibrate the weighting scheme with the sample size or the context. Any weights satisfying Assumption \ref{assum:weights} will result in a pointwise valid and consistent test. However, their choice may direct power in different directions for small samples. In our simulations and empirical applications we always use $w_j = \min\{1,(1/2)^{j-J_0}\}$ for $J_0=20$ with $J_0=30$ as  a robustness check in the appendix. This scheme was chosen because it performed well in simulations.
\end{remark}

 Lemma \ref{lem:cj_Pi} uses the Singular Value Decomposition from Example 1 of \cite{CarrascoPaper} and \cite{faridani2025testingunderpoweredliteratures} to rewrite $c_j(\Pi)$. This also reveals that the $c_j(\Pi)$ must decay exponentially fast to zero in $j$.

\begin{lemma}\label{lem:cj_Pi}
Let Assumption \ref{assum:cts_TR} hold. Let $\chi_j(t)\equiv \psi_j(t/\sqrt{2})$ and $\varphi_2(t)\equiv \varphi(t/\sqrt{2})/\sqrt{2}$. For any probability distribution $\Pi$ (discrete, continuous, or otherwise), if $H\sim \Pi$, then:
    $$c_j(\Pi) = \frac{1}{2^{j/2}}\mathbb{E}_\Pi\left[\chi_j(H)\varphi_{2}(H)\right]$$ 

    Proof: Section \ref{proof:lem:cj_Pi}
\end{lemma}

Lemma \ref{lem:cj_Pi} lets us define the vectors  $ \mathbf{b}_{x}\in \ell^2$ as the $\mathbf{c}(\Pi)$ for a $\Pi$ that has a single mass point at $x$:
\begin{align}
     {b}_{x,j} &\equiv \frac{1}{2^{j/2}} \chi_j(x)\varphi_{2}(x)=c_j(\delta_x)
\end{align}

\noindent Define $W \equiv \text{diag}(w)$ and $\Lambda_w$ as the closure of the convex hull of all of these $\mathbf{b}_x$:
\begin{align}
         \Lambda_w &\equiv \overline{\text{conv}\left(W\mathbf{b}_x \: : \: x \in \mathbb{R}\right)}\subset \ell^2
\end{align}

So we can rewrite the distance function as a projection in $\mathcal{\ell}^2$:
 \begin{align}\label{eq:projform_w}
     {d}_w(\theta) &= \min_{\lambda \in \Lambda_w}|| W\theta - \lambda||_{\ell^2}
 \end{align}

We can extend the testable implication of Theorem \ref{thm:testable_implication} and Lemma \ref{lem:ynorm_infsum} to the weighted distance as well. Lemma \ref{lem:H0_dw} formalizes this. 

\begin{lemma}\label{lem:H0_dw}
    Let Assumptions \ref{assum:cts_TR} and \ref{assum:weights} hold. Then: $$\mathbf{H}_0 \iff   {d}_w(\theta_0)=0$$ Proof: Section \ref{proof:lem:H0_dw}.
\end{lemma}
 Next we verify that $\Lambda_w$ is compact and convex.

\begin{lemma}\label{lem:compact_Lambda}
  If Assumptions \ref{assum:cts_TR} and \ref{assum:weights} hold,  $\Lambda_w$ is compact and convex in $\ell^2$. Proof: Section \ref{proof:lem:compact_Lambda}
\end{lemma}

We are ready to specify the sampling process and the test statistic. Ideally, the meta-analyst would like to plug the true $\theta_0$ into $ {d}_w$. This is impossible because the meta-analyst does not observe the full population of $T \mid R$. Instead, the meta-analyst has only a sample of $n$ $t$-statistics $\{t_i\}_{i=1}^n$. Rather than the vector of expectations $\theta_0$ they observe the analogous  vector of sample means $\widehat{\theta}_n$: 
\begin{equation}
 \widehat{\theta}_{nj} =\frac{1}{n}\sum_{i=1}^n\varphi(t_i)\psi_j(t_i)
\end{equation}

\noindent We will assume that the $n$ $t$-statistics are reported by $m$ articles. While $t$-statistics are dependent within an article, they must be independent across articles and the number of $t$-statistics per article must be uniformly bounded. Assumption \ref{assum:articles} formalizes this clustered sampling setup. 
\begin{assumption}\label{assum:articles}
The following all hold:
    \begin{itemize}
        \item The sample is composed of $m$ different articles indexed by $g$ each reporting $N_g$ $t$-statistics.
        \item Article-level data $(N_g, \{T_i\}_{i \in g})$ is independent and identically distributed across articles.
        \item There exists a universal constant $c_0>0$ such that $1\leq N_g\leq c_0$.
    \end{itemize}
\end{assumption}

Let $\widehat{\theta}_n^*$ be the vector of sample means computed from the bootstrapped sample (where resampling is article-by-article). Lemma \ref{lem:clt_and_bootstrap} confirms that the sample means follow a CLT and that the bootstrap satisfies Assumption 3 of \cite{FangSantos}. 

\begin{lemma}\label{lem:clt_and_bootstrap}
    Let Assumptions   \ref{assum:weights} and \ref{assum:articles} hold. Then $\sqrt{m}\,(W\widehat{\theta}_n -W\theta_0) \xrightarrow{d} \mathbb{G}_0$ in $\ell^2$, where $\mathbb{G}_0$ is a tight, mean-zero Gaussian random element of $\ell^2$, and
$$
  \sup_{f\in BL_1(\ell^2)}\left|\,E\left[f\left(\sqrt{m}\,(W\widehat{\theta}_n^* -W\widehat{\theta}_n)\right)\,\mid\,\left\{t_i\right\}_{i=1}^n,\left\{N_g\right\}_{g=1}^m\right]-E\left[ f(\mathbb{G}_0)\right]\right| =o_p(1)
$$
where $BL_1(\ell^2)$ is the set of functions $f:\ell^2\to[-1,1]$ with $|f(a)-f(b)|\le\|a-b\|_{\ell^2}$.

    Proof: Section \ref{proof:lem:clt_and_bootstrap}
\end{lemma}

Now we pass the asymptotic results for $\theta$ through the function ${d}_w(\theta )$. The technical difficulty is that we cannot apply the usual delta method to get the asymptotic distribution of the test statistic $\sqrt{m}{d}_w(\widehat{\theta}_n)$ because under the null hypothesis, the function ${d}_w(\theta)$ is not {\it fully} differentiable at $\theta_0$. This is because the directional derivative is zero in some directions (that point into the set of honest $t$-curves) and positive in others (that point away from the set of honest $t$-curves). 

 Fortunately, \cite{FangSantos} have already shown how to derive a valid bootstrap in just this situation. Their Proposition 4.1  shows that ${d}_w(\theta)$ is directionally Hadamard differentiable even under the null hypothesis.\footnote{\cite{FangSantos} attribute the original version of this result to \cite{ZARANTONELLO1971237}.} Equation (\ref{eq:projform_w}) has already expressed our test statistic in their notation. They show that the asymptotic distribution of the residual of a projection of a vector plus noise onto a convex set is the residual of a projection of the noise onto the {\it tangent cone} $TC_{\theta}$ evaluated at the true projection point $\theta_0$.  The tangent cone of the convex set $\Lambda_w$ at the point $W\theta \in \Lambda_w$ is defined as the closure of the set of all vectors that point from $W\theta$ into the set $\Lambda_w$:
 \begin{equation}
     TC_{\theta} \equiv \text{cl}\{ \mathbf{r}=\alpha(\mathbf{u}-W\theta)\text{ for }\alpha \geq 0, \: \mathbf{u}\in \Lambda_w \}
 \end{equation}
\noindent The norm of the projection residual $ \phi'_{\theta_0}(\mathbf{e})$ of a vector $\mathbf{e}$ onto the tangent cone $ TC_{\theta_0}$ is defined as:
\begin{align}
    \phi'_{\theta_0}(\mathbf{e})&\equiv \inf_{\mathbf{r}\in TC_{\theta_0}}\sqrt{\sum_{j=0}^\infty (e_j-r_j)^2  }
\end{align}
\noindent Intuitively, the tangent cone captures the directions in which $\widehat{\theta}_n$ can move while remaining consistent with the null. The test statistic is large when $\widehat{\theta}_n$ moves in a direction that points away from the null.
 
With these definitions in hand, applying Proposition 4.1 of \cite{FangSantos} shows that the asymptotic distribution of the test statistic is the distribution of the Gaussian $\mathbb{G}_0$ projected onto the tangent cone:
\begin{equation}
    W\theta_0 \in \Lambda_w \implies  \sqrt{m}{d}_w(\widehat{\theta}_n)\xrightarrow{d} \phi'_{\theta_0}(\mathbb{G}_0)
\end{equation}

 To estimate the directional derivative $\phi'_{\theta_0}$ we will use the numerical estimator $\widehat{\phi}'_{n,s_n}$ from Equation (25) of \cite{FangSantos}:
\begin{align}
    \widehat{\phi}'_{n,s_n}(h) \equiv \frac{1}{s_n}\left({d}_w(\widehat{\theta}_n +s_n h)-{d}_w(\widehat{\theta}_n)\right)
\end{align}
\noindent where $s_n = n^{-1/3}$. Then as a consequence of Theorem 3.2 of \cite{FangSantos}:
\begin{align}
    \widehat{\phi}'_{n,s_n}(\sqrt{m}(\widehat{\theta}_n^*-\widehat{\theta}_n)) \xrightarrow{d} \phi'_{\theta_0}(\mathbb{G}_0)
\end{align}


This motivates the following critical value for the projection test.  If $F^*_n$ is the bootstrapped CDF of $\widehat{\phi}_{n,s_n}'(\sqrt{m}(\widehat{\theta}_n^*-\widehat{\theta}_n))$, then the critical value for the projection test of size $\alpha \in (0,1)$ is:
\begin{align}
    \text{cv}(\alpha) \equiv F_n^{*,-1}\left(1-\alpha\right) 
\end{align}

Theorem \ref{thm:validity_and_consistency_ideal} formalizes the logic sketched above and guarantees that the test is valid pointwise and consistent against all fixed alternatives. The test statistic $\sqrt{m}{d}_w(\widehat{\theta}_n )$ does not require any tuning parameters that change with the sample size.


\begin{theorem}\label{thm:validity_and_consistency_ideal}
    Let Assumptions \ref{assum:cts_TR}, \ref{assum:weights}, and \ref{assum:articles} hold. Let $\eta_m$ be any sequence of numbers such that $\eta_m=O(1)$ and $\liminf_{m\to\infty}\eta_m>0$ (e.g. any constant). Let $\overline{cv}(\alpha)\equiv cv(\alpha)+\eta_m$ and let $\sqrt{m}s_n\to \infty$ and $s_n\to 0$. Then both of the following also hold:

\begin{enumerate}
    \item  {\bf Validity:}   
$\mathbf{H}_0 \implies \limsup_{m\to \infty}\mathbb{P}\left[\sqrt{m}{d}_w(\widehat{\theta}_n ) > \overline{cv}(\alpha)\right]\leq \alpha$
\item    {\bf Consistency:}  Fix any alternative distribution of $T \mid R$ such that $\mathbf{H}_0$ is false. Then:
    $$\lim_{m\to \infty}\mathbb{P}\left[\sqrt{m}{d}_w(\widehat{\theta}_n ) > \overline{cv}(\alpha)\right]=1$$
\end{enumerate}

Proof: Section \ref{proof:thm:validity_and_consistency_ideal}

\end{theorem}

\begin{remark}\normalfont
    When the support of $\Pi_0$ contains an interval, the tangent cone can be dense in $\mathcal{\ell}_2$ which makes $\phi'_{\theta_0}=0$. This is not a defect of this test because local power must be size in all directions for any locally valid test of ${\mathbf{H}_0}$ (Appendix \ref{app:local_untestability}). A simple way to preserve validity in such cases is to add a tiny number $\eta_m$ to the critical value in order to prevent it from converging to zero. A natural value for $\eta_m$ is simply the maximum numerical approximation error $\mathcal{E}$ incurred by the projection implementation. Numerical implementation is formalized in the next section.
\end{remark}

\section{Numerical Implementation}

In practice, the projection distance $d_w$ cannot be computed exactly since it is an infimum over all discrete distributions of an infinite sum. Fortunately, it can be numerically approximated with arbitrarily good precision. In this section we demonstrate how to do so while preserving pointwise validity and consistency.

Since the $w_j$ must eventually decay, we need only compute $\theta_j$ up to some large number $J$. The researcher can make $J$ as large as is computationally feasible. Define the truncated distance as:
\begin{align}\label{eq:Dj_J}
    d_{w,J}(\theta)\equiv\inf_{\Pi\in\mathcal{D}}\sqrt{\sum_{j=0}^J w_j^2\big(\theta_j-c_j(\Pi)\big)^2}
\end{align}

 Taking the infimum in (\ref{eq:Dj_definition}) or (\ref{eq:Dj_J}) over the set of all discrete distributions $\mathcal{D}$ is computationally challenging because the $\Pi$ are high dimensional. Instead we will minimize over a subset of distributions and then bound the resulting approximation error. Specifically, consider the subset $ {\mathcal{D}}_{\mathcal{X}}  \subset \mathcal{D}$ where the distributions have support contained within a chosen finite set $ \mathcal{X}\subset \mathbb{R}$. 
 \begin{align}
 {\mathcal{D}}_{\mathcal{X}} &\equiv \{\Pi \in \mathcal{D}\text{ s.t. $\Pi$ has support contained in $\mathcal{X}$}\}
\end{align}

 The resulting approximate distance $\widetilde{d}_{w,J}(\theta)$ is easy to compute in practice because it is the minimum over a closed and convex set:
\begin{align}
       \widetilde{d}_{w,J}(\theta) &\equiv \min_{\Pi \in {\mathcal{D}}_{\mathcal{X}}}\sqrt{\sum_{j=0}^J  w_j^2\left(\theta_j-c_j(\Pi)\right)^2 }
\end{align}
\noindent Analogously, let $\widetilde{cv}(\alpha)$ denote the $1-\alpha$ quantile of the bootstrapped distribution of $\widetilde{\phi}'_{n,s_n,J}(\sqrt{m}(\widehat{\theta}_n^*-\widehat{\theta}_n))$ where $\widetilde{\phi}'_{n,s_n,J}(h) \equiv \frac{1}{s_n}\left(\widetilde{d}_{w,J}(\widehat{\theta}_n +s_n h)-\widetilde{d}_{w,J}(\widehat{\theta}_n)\right)$. 

 Projecting onto ${\mathcal{D}}_{\mathcal{X}}$ instead of $\mathcal{D}$ incurs some approximation error that can be forced to be small by using a grid $\mathcal{X}$ wide and fine enough. Lemma \ref{lem:max_grid_approx_error} tells us how wide and fine the grid must be in order to control the approximation error. This holds uniformly across $J\in \mathbb{N}$.

\begin{lemma}\label{lem:max_grid_approx_error}
Let the weights satisfy $0\leq w_j\leq 1$. Let $\mathcal{X} \subset [-L,L]$ be a grid of $n_x$ evenly spaced points with spacing $\delta = \frac{2L}{n_x-1}$. Define   $x(y)$ to be the closest member of $\mathcal{X}$ to y, define $\mathcal{E}(L,\delta)\equiv\mathcal{E}_1(L)+\mathcal{E}_2(L,\delta)$, and:
    \begin{align*}
\mathcal{E}_1\left(L\right) &\equiv \sqrt{\int_{-\infty}^\infty \varphi(t)\varphi(t-L)^2dt}\\
\mathcal{E}_2\left(L,\delta\right) &\equiv \sup_{y\in [-L,L]}\sqrt{\int_{-\infty}^\infty \varphi(t)\left(\varphi(t-y)-\varphi(t-x(y))\right)^2dt }
\end{align*}

Then for every $\theta, J$: \begin{align*}
  \widetilde{d}_{w,J}(\theta)&\leq d_{w,J}(\theta)+ \max\{\mathcal{E}_1(L),\mathcal{E}_2(L,\delta)\} \\
  \widetilde{d}_{w,J}(\theta)&<d_w(\theta)+\mathcal{E}(L,\delta)
\end{align*}

Proof: Section \ref{proof:lem:max_grid_approx_error}.
\end{lemma}

\begin{remark} \normalfont
    In our simulations and empirical application we set $\mathcal{X}$ to be a grid of 3000 evenly spaced points between $-6.5$ and $6.5$ plus points at zero and 9999. This yields $\mathcal{E}\left(6.5,\frac{13}{2999}\right)  = 0.00063$. This is small compared to critical values. We then set $J=40$. Truncation only makes distances smaller and cannot affect validity and cutting off at $40$ makes the sum of all squared weights beyond $J$ much less than $\mathcal{E}\left(6.5,\frac{13}{2999}\right)$ for our preferred weights.
\end{remark}

First we confirm that fixed grids work and then extend to asymptotically growing ones. Proposition \ref{prop:implementation} shows that the test remains pointwise valid for any fixed numerical implementation. Similarly, any alternative can also be detected with a fine enough implementation.

\begin{proposition}\label{prop:implementation}
Let Assumptions \ref{assum:cts_TR}, \ref{assum:weights}, and \ref{assum:articles} hold. Fix $J$ and the grid $\mathcal{X}$. Then both of the following hold:
\begin{enumerate}
    \item {\bf Validity:} $\mathbf{H}_0\implies\lim_{m\to\infty}\mathbb{P}\left[\sqrt{m}\,\widetilde{d}_{w,J}(\widehat{\theta}_n)>\widetilde{cv}(\alpha)+\sqrt{m}\,\mathcal{E}(L,\delta)\right]=0$.
    \item {\bf Consistency:} Fix any distribution of $T\mid R$ such that $d_{w,J}(\theta_0)>\mathcal{E}(L,\delta)$. Then
    \[
      \lim_{m\to\infty}\mathbb{P}\left[\sqrt{m}\,\widetilde{d}_{w,J}(\widehat{\theta}_n)>\widetilde{cv}(\alpha)+\sqrt{m}\,\mathcal{E}(L,\delta)\right]=1 .
    \]
\end{enumerate}

Proof: Section \ref{proof:prop:implementation}
\end{proposition}

Finally, to make the test both fully consistent and feasible to run, the researcher can progressively tighten the numerical implementation with the sample size. Theorem \ref{thm:implementation_asymptotic} shows that when done at the right rate, the resulting feasible test is pointwise valid and consistent against all fixed alternatives.

\begin{theorem}\label{thm:implementation_asymptotic}
Let Assumptions \ref{assum:cts_TR}, \ref{assum:weights}, and \ref{assum:articles} hold. Let the truncation $J=J_m$ and the grid parameters $L=L_m$ and $\delta=\delta_m$ vary with $m$ so that $J_m\to\infty$, $\mathcal{E}_m\equiv\mathcal{E}(L_m,\delta_m)\to0$, and
\[
  \sqrt{m}\,\min\big\{\mathcal{E}_1(L_m),\,\mathcal{E}_2(L_m,\delta_m)\big\}\to\infty .
\]
Then both of the following hold:
\begin{enumerate}
    \item {\bf Validity:} $\mathbf{H}_0\implies\lim_{m\to\infty}\mathbb{P}\left[\sqrt{m}\,\widetilde{d}_{w,J_m}(\widehat{\theta}_n)>\widetilde{cv}(\alpha)+\sqrt{m}\,\mathcal{E}_m\right]=0$.
    \item {\bf Consistency:} Fix any distribution of $T\mid R$ such that $\mathbf{H}_0$ is false. Then
    \[
      \lim_{m\to\infty}\mathbb{P}\left[\sqrt{m}\,\widetilde{d}_{w,J_m}(\widehat{\theta}_n)>\widetilde{cv}(\alpha)+\sqrt{m}\,\mathcal{E}_m\right]=1 .
    \]
\end{enumerate}

Proof: Section \ref{proof:thm:implementation_asymptotic}
\end{theorem}

\begin{remark}\normalfont
    Here the deterministic numerical approximation error $\mathcal{E}_m$ plays a similar role to the sequence $\eta_m$ from the idealized test in Theorem \ref{thm:validity_and_consistency_ideal}.
\end{remark}

\subsection{Computation in Practice}

The computational task of computing this projection distance is well-conditioned even though the problem of actually identifying the minimizing distribution $\Pi$ is ill-conditioned. We compute projections with the \verb|nnls| package in R and set $J=40$ in our empirical application and simulations.\footnote{Our code scales the test statistic, its bootstrap draws, and the term $\mathcal{E}_m$ by
$\sqrt{n}$ rather than $\sqrt{m}$. All three are multiplied by the same
factor $\sqrt{n/m}$, so the resulting decisions, $p$-values, and $\widehat{B}$ (see Section \ref{sec:misspecification}) are exactly
those of the procedure in the text with $s_n\sqrt{n/m}$ in
place of $s_n$.}  The feasible test statistic  $\widetilde{d}_{w,J}(\theta)$ can be rewritten as an explicit quadratic programming problem that is legible to standard packages:
\begin{align}
       \widetilde{d}_{w,J}(\theta) &= \min_{\mathbf{a} \in [0,1]^{|\mathcal{X}|}\text{ s.t. }\mathbf{a}'\mathbf{1}=1}\sqrt{\sum_{j=0}^J  w_j^2\left(\theta_j-\sum_{x \in \mathcal{X}}\frac{a_{x} }{2^{j/2}} \chi_j(x)\varphi_{2}(x)\right)^2 }
 \end{align}
\begin{remark}\label{rem:symmetrization}\normalfont
  In practice we do not always observe the sign of every $t$-score. Accounting for this,  \cite{Elliott}, \cite{Simonsohn}, \cite{kudrinjmp24}, and others study the  $p$-curve of the two-sided $p$-value which does not depend on the sign of $T$.  In practice we recommend doing the same by symmetrizing the $t$-statistics. $t$-statistics of two-sided t-tests are often reported as an absolute value or after a sign flip and papers may report a mix of one-sided and two-sided $t$-statistics. To account for this, take the sample of $t$-statistics $\{t_i\}_{i=1}^n$ and replace it with $\{|t_i|,-|t_i|\}_{i=1}^n$. Under the null hypothesis this is equivalent to replacing  $\Pi$ with its symmetrized version. When resampling with the bootstrap, first resample from the original dataset $\{t_i\}_{i=1}^n$ and then symmetrize. We recommend symmetrization as a default in every application.
\end{remark}

\begin{remark}\label{rem:shifting}\normalfont
 The Cox-Shi tests of \cite{Elliott}, the $T_H$ test of \cite{kudrinjmp24}, and others test the $p$-curve only for low $p$-values which directs power towards places near significance thresholds where selection is most salient. We recommend doing the same. To do this, we recommend shifting all the $t$-statistics by 1.96 or another relevant critical value. This is because the norm $||\cdot||_T$ weights down distortions in $T\mid R$ far from zero. But we expect distortions to be near $1.96$. Therefore, power may increase if the meta-analyst simply subtracts $1.96$ from each $t$-statistic in the dataset (after symmetrization). We recommend this translation as a default in every application. 
\end{remark}

\section{Avoiding Spurious Rejections}\label{sec:misspecification}

The projection test can potentially detect small distortions in the $t$-curve and the meta-analyst might be concerned that it will instead pick up small benign distortions in the $t$-curve. In this section we will show how to rule these out.

In order to interpret a rejection of ${\mathbf{H}_0}$ as evidence of $p$-hacking, we need to implicitly assume that $T|H \sim N(H,1)$ exactly in the absence of selective reporting. In practice, this normality assumption can be slightly violated. We wish to reject the null only when the distortions in the $t$-curve were severe enough not to be attributable to small amounts of non-normality. We consider here the main cause of benign violations: the $t$-statistic was reported after rounding and then  possibly de-rounded by the meta-analyst. This can only change an honest $t$-curve by a small amount. So by augmenting the critical value, the meta-analyst can avoid spurious rejections. The ease of this augmentation is one of the main advantages of the projection test.

\subsection{The Weakened Null Hypothesis}

Rejections of $\mathbf{H}_0$ may be spurious when the $t$-curve is slightly distorted by innocuous factors unrelated to selective reporting. To avoid this, we can weaken the null hypothesis---making it harder to reject---by allowing for distortions in the $t$-curve commensurate with the noise term in the $t$-statistic having an unknown PDF $g$ that is close—--but not equal—--to the standard normal. Suppose that the meta-analyst knows that the total difference between the two densities is upper bounded by some constant $\Delta \geq 0$. The augmented null hypothesis $\mathbf{H}_0^{\Delta} $ formalized below states that the distribution of $T \mid R$ can be explained by a true effect $H$ plus a noise term $U$ that is within $\Delta$ of a standard normal:

\begin{definition}
  Given a chosen number $\Delta \geq 0$, the weakened null hypothesis  $\mathbf{H}_0^{\Delta} $ states that there exist $\Pi \in \mathcal{P}$ and a function $g(t,h)$ such that:
 \begin{enumerate}
     \item  $f_{T \mid R}(t) = \int_{-\infty}^\infty g(t,h)d\Pi(h)$
     \item $g(t,h)\geq 0$ and $\int_{-\infty}^\infty g(t,h)dt = 1$ for all $h\in \mathbb{R}$
     \item $\sup_{h\in\mathbb{R}}||g(t,h)-\varphi(t-h)||_T \leq \Delta $
 \end{enumerate}
\end{definition}

\noindent So if $\Delta = 0$, then $\mathbf{H}_0^{\Delta} $  agrees with $\mathbf{H}_0 $. Otherwise, $\mathbf{H}_0^{\Delta} $  is weaker. 

 We test $\mathbf{H}_0^{\Delta}$ using the same test statistic as before: $\sqrt{m}\widetilde{d}_w(\widehat{\theta}_n )$. This null implies that $d_w(\theta_0)\leq \Delta$ since $w_j\leq 1$. The only difference is that we now add $\sqrt{m}\Delta$ to the critical value.
 
 When the researcher is unwilling to commit to a specific value of $\Delta$, they can instead report the {\it breakdown statistic}, or the smallest value $\widehat{B}$ of $\Delta$ that would overturn their detection of selective reporting:
\begin{equation}
    \widehat{B} \equiv \max\{0,\widetilde{d}_{w,J}(\widehat{\theta}_n )- \widetilde{cv}(\alpha)/\sqrt{m}-\mathcal{E}_m\}
\end{equation}

\noindent $\widehat{B}$
  is the smallest deviation from normality that would be sufficient to explain away the observed lack of smoothness in the $t$-curve. The following  subsections discuss the leading case in meta-analysis: rounding and de-rounding.


\subsection{Rounding and de-rounding}\label{sec:rounding}

Empirical papers almost always report rounded versions of reported statistics. This creates bunching---particularly around $t=2$---that can lead many tests for selective reporting to spuriously reject. The meta-analysis literature usually (but not always) uses derounding procedures, e.g. adding uniform noise at the level of the smallest significant digit, to smooth out such heaps \cite{bb,Elliott}. Derounding procedures are usually stochastic, posing problems for replicability. The details of the procedures can vary, there is little theory to back them up when standard errors are rounded, and it is often necessary to drop data points that are so rounded as to be uninformative. The debate over how to approach the rounding problem is therefore active and first-order \citep{KranzPutz2022,bb_reply}. \cite{elliott2024powertestsdetectingphacking} note that derounding makes a large difference in empirical results.

The projection test offers a path forward. It has two key advantages. First, we propose here a deterministic derounding procedure under which the critical value of the projection test needs only a very small extra allowance to remain valid. Our procedure works when either $t$-scores themselves are rounded or when they are computed as the ratio of rounded effects and rounded standard errors. Since our derounding is deterministic, there are no replicability issues or extended computation times. Adding uniform noise is optional and our allowance remains valid either way.

The second advantage of the projection test is that it is by nature insensitive to rounding in the first place. Unlike the tests of \cite{Elliott} it does not direct its power specifically towards finding roughness or wrong derivatives in the $t$-curve. So it is much less sensitive to spurious roughness caused by rounding. This structure has two practical implications: (a) the projection test's results in our empirical application and simulations barely change when the meta-analyst skips derounding or uses a procedure other than ours and (b) the derounding allowance for $\widehat{B}$ turns out to be tiny and the resulting power loss is therefore minimal. Taken together, these advantages make the projection test particularly attractive under rounding. 

\subsection{Deterministic derounding for the projection test}\label{sec:der_procedure}

To understand our procedure, first let the {\it rounding unit} $\kappa$ for each observation of $T$ be the width of the interval of values that round to the reported $T$. Rounding can happen in two ways:
\begin{enumerate}
    \item When the $t$ was printed directly, use it as printed. If it is rounded to the $a$ place, set the rounding unit $\kappa = a$. For example if $t=1.5$, $\kappa = 0.1$. 
     \item Suppose instead that $t$ is computed by the meta-analyst as the ratio $\frac{m}{s}$ of the rounded coefficient $m$ and the rounded standard error $s$.  $m$ is rounded to the $a$ place and $s$ rounded to the $b$ place. First, set $s_{min}\equiv s-\frac{b}{2}$ or the smallest value that $s$ could have  been before rounding. Then the derounded $t_{der} = \frac{m}{s_{min}}$. In other words, leave the coefficient as is and shrink the standard error. Set $\kappa = a/s_{min}$. For example, if $m=2.0$ and $s=0.50$, $\kappa = 0.1/0.495$.
\end{enumerate}
\noindent Notice that this derounding procedure adds no noise,  shrinks standard errors, and does nothing to $t$-scores directly printed or to coefficients. We will show later that adding noise to the numerator is purely optional. 

\subsection{Validity of projection test under derounding}

Now we show how to use the breakdown statistic to keep the projection test valid under rounding and derounding. We do this first for case 1 where $t$-scores are printed directly and second for case 2 where $t$-scores are ratios of rounded statistics. The key assumption is that the level of rounding does not depend on the researcher's sampling error conditional on the true effect. We will also need an additional decay condition on the weights which our preferred weights still satisfy.

First some additional notation. $\xi_j(t)$ denotes the $j$th coefficient function, so a statistic $T$ has coefficients $\theta_j=E[\xi_j(T)]$. If we have not symmetrized and shifted (as in Remarks \ref{rem:symmetrization} and \ref{rem:shifting}), then $\xi_j(t)\equiv \varphi(t)\psi_j(t)$. If we did shift and symmetrize, then $\xi_j(t)\equiv \frac{1}{2}\left(\varphi(t-1.96)\psi_j(t-1.96)+\varphi(-t-1.96)\psi_j(-t-1.96)\right)$. Let $\xi(t)$ denote the function that maps to the entire vector.

Proposition \ref{prop:rounding} says that when $t$-scores are rounded directly, then we can keep the projection test valid by rejecting ${\mathbf{H}_0}$ only when $\widehat{B}$ exceeds $C\,E[\kappa^6]$ for fixed known constant $C$ that can be computed numerically and equals $0.13$ for our preferred parameterization. This result is unchanged whether or not the meta-analyst added derounding noise (no noise means setting $\mathbb{P}[U=0]=1$). 

\begin{proposition}\label{prop:rounding}
Assume that weights satisfy $\sum_{j=0}^\infty w_j^2 (j+1)^6<\infty$. Suppose that $T=H+Z$, with $Z\sim N(0,1)$ and $Z\independent H$. There is no selective reporting, but each study reports $T$ rounded to the nearest multiple of $\kappa$ possibly plus independent noise or: $\kappa\lfloor T/\kappa+\frac{1}{2}  \rfloor  + \kappa U$ where $U$ is independent of $(T,H,Z,\kappa)$. $\kappa$ is random but must satisfy $\kappa \independent Z \:|\: H$. Then $d_w(\mathbb{E}[\xi(\kappa\lfloor T/\kappa+\frac{1}{2}  \rfloor  + \kappa U)])\le C\,E[\kappa^6]$ for some fixed constant $C >0$. 

If, in addition, $w_j = \min\{1,(1/2)^{j-J_0}\}$ then $C=0.13$ for $J_0=20$ and $C=0.31$ for $J_0=30$ when $t$-scores are shifted and symmetrized and $C=0.17$ for $J_0=20$ and $C=0.43$ for $J_0=30$ without shifting and symmetrization. Proof: Section \ref{proof:prop:rounding}.
\end{proposition}

When $t$-scores are instead ratios of rounded coefficients and standard errors, shrinking the standard errors to their smallest possible pre-rounded value allows us to use the exact same breakdown statistic as before.  Proposition \ref{prop:ratio} formalizes this. Independent noise is again optional because $\mathbb{P}[U=0]=1$ is allowed.

\begin{proposition}\label{prop:ratio}
Assume that weights satisfy $\sum_{j=0}^\infty w_j^2 (j+1)^6<\infty$. Suppose that  a study's estimate $\hat\mu$ and standard error $\hat\sigma>0$ satisfy $T=\hat\mu/\hat\sigma=H+Z$, with $Z\sim N(0,1)$ and $Z\independent H$. There is no selective reporting but $(\hat\mu, \hat\sigma)$ are reported as rounded $(m,s)$ satisfying: $m=a\lfloor \hat{\mu}/a+\frac{1}{2}  \rfloor$ and $s=b\lfloor \hat{\sigma}/b+\frac{1}{2}  \rfloor>0$. Assume that $(a,b,\hat{\sigma})\independent Z|H$ (i.e. noise comes from the numerator $\hat\mu$). Let $s_{\min}\equiv s-b/2$ and let $\kappa=a/s_{\min}$. Then:
\begin{enumerate}
\item[(i)] $\tilde{T}\equiv\hat\mu/s_{\min}$ satisfies ${\mathbf{H}_0}$.
\item[(ii)] $T_{der}\equiv (m+aU)/s_{\min}=\kappa\lfloor \tilde{T}/\kappa+\frac{1}{2}  \rfloor  +  \kappa U$  with $U\independent(T,H,Z,\hat{\sigma},a,b)$ satisfies $d_w(\mathbb{E}[\xi(T_{der})])\le C\,E[\kappa^6]$, with $C$ as in Proposition \ref{prop:rounding}.
\end{enumerate}
Proof: Section \ref{proof:prop:ratio}.
\end{proposition}

A small number of very coarsely rounded statistics can drive $\mathbb{E}[\kappa^6]$. We recommend dropping observations with $\kappa > 0.5$. This means dropping $t$-scores whose true values cannot be determined to within $0.5$ which is very coarse. Dropping is valid when the rounding units are conditionally independent of the noise variable $Z$ or $\kappa \independent Z |H$.  This practice has a precedent in \cite{KranzPutz2022} who recommend dropping based on the significand of the standard error.  In our application our rule results in dropping 7.1\% which is much less than the 36.6\% their rule recommends in our subsample.  

This results in the following recipe for practitioners: (1) Shrink the standard errors to their smallest possible pre-rounded value. (2) Drop observations $i$ where $\kappa_i > 0.5$. (3) Symmetrize  and shift as in Remarks \ref{rem:symmetrization} and \ref{rem:shifting}. (4) Run the projection test with $w_j = \min\{1,(1/2)^{j-20}\}$. (5) Reject ${\mathbf{H}_0}$ if $\widehat{B}>0.13\frac{1}{n}\sum_{i=1}^n\kappa_i^6$. Adding uniform noise to the $t$-score  or its numerator is optional and the procedure maintains validity either way.

\subsection{Further Robustness Considerations}

Besides rounding and de-rounding, the $t$-curve could also be distorted if a large fraction of the $t$-statistics are invalid---i.e. they do not have a conditional distribution within $\Delta$ of the normal. In this case $\mathbf{H}_0^{\Delta}$ is violated without selective reporting. The possibility of rejections coming from many invalid $t$-tests is a shortcoming of every test for $p$-hacking, including those of \cite{Elliott}, \cite{bb}, and others. Nevertheless, the presence of many invalid tests would itself be at least as important a finding as selective reporting---and likely has similar implications for the credibility of empirical research.

\section{Empirical Application}\label{sec:application}


This section applies the projection test to two empirical settings: one where we detect substantial evidence of selective reporting and one where we do not. The first application finds that the projection test (1) finds significant distortions in $t$-curves of top economics publications where the tests of \cite{Elliott} do not after adding uniform noise and (2) demonstrates robustness of our projection test to rounding and derounding. Our second empirical exercise studies a big sample of large medical trials from systematic reviews (published and unpublished) and demonstrates that the projection test does not simply find $p$-hacking everywhere. 

The importance of derounding is well documented and remains an active debate \citep{KranzPutz2022,bb_reply,elliott2024powertestsdetectingphacking}. We provide results for four different methods of derounding: (1) no derounding, (2) add uniform noise to the numerator and denominator at the level of the last significant digit and take the mean $p$-value over 200 repetitions, (3) leave the numerator unchanged and replace the denominator with its smallest possible pre-rounded value as in Section \ref{sec:der_procedure} and (4) treat the numerator as in 2 and the denominator as in 3.

The first application is the data from \cite{bb}. This data consists of 21,740 two-sided t-statistics drawn from 684 articles. These articles were the universe of articles published in 25 top economics journals in 2015 and 2018 that used one of four causal methods. Each $t$-statistic corresponds to a two-sided $t$-test of a main hypothesis of interest. The coefficients in the numerators of the $t$-statistics were estimated using Randomized Controlled Trials (RCT), Difference in Differences (DID), Regression Discontinuity Designs (RDD), or Instrumental Variables (IV). The original authors state that all  regression controls, constant terms, balance and robustness checks, heterogeneity of effects, and placebo tests were excluded. Out of an original sample of 21,740 tests we drop those for which the coefficient or standard error is missing, the standard error is printed as zero, or $\kappa > 0.5$ (see Section \ref{sec:der_procedure}).\footnote{We calculate the rounding unit of the coefficient in the \cite{bb} data to be the finest place value among its table’s coefficients, but no finer than its standard error’s rounding unit.} This leaves 19,155 tests in 641 articles for our analysis of the \cite{bb} data.\footnote{Unlike other meta-studies, we leave in tests where $z=0$ exactly if enough significant digits are recorded.}

The second application is large medical trials. The data are an extraction of the Cochrane Database of Systematic Reviews (CDSR) collected by \cite{Schwab2020} and analyzed by \cite{vanZwet2021}. Cochrane's database collects systematic reviews which attempt to include unpublished and less-known results in order to mitigate the file drawer problem. Their original sample contains 411,984 outcomes tested in 67,927 medical studies. We take the subsample consisting only of RCTs with at least 100 participants, where the outcome is a continuous efficacy measure with arm-level means and standard deviations available. We remove duplicated rows and reconstruct test statistics from arm-level means and sds.  We drop 4456 heavily rounded outcomes corresponding to $\kappa = \frac{2a}{\text{se}_{min}} > 0.5$ where $a$ is the smallest rounding unit of the outcome in treatment and control groups.\footnote{The numerator of $\kappa$ is $2a$ here because here rounding was done by treatment and control groups separately.} Our final subsample consists of 17,680 outcomes reported in 5,310 large medical trials collected by Cochrane. This makes the effective sample size many times larger for the Cochrane trials than for the economics studies.

Table \ref{tab:rounding_ekw_keepmu0_drop_se_mi_cochrane} shows the results of both applications together. The first four panels show broad evidence of significant distortions in the $t$-curve in top economics publications. The projection test rejects at the 5\% level in every subsample and for every derounding method. Figure \ref{fig:Figure_tcurve_shape_smin_noise_keepmu0_drop_se_mi} visualizes why the projection test rejects by plotting the histogram of $t$-scores against the closest honest $t$-curve.

\begin{figure}
    \centering
    \includegraphics[width=\linewidth]{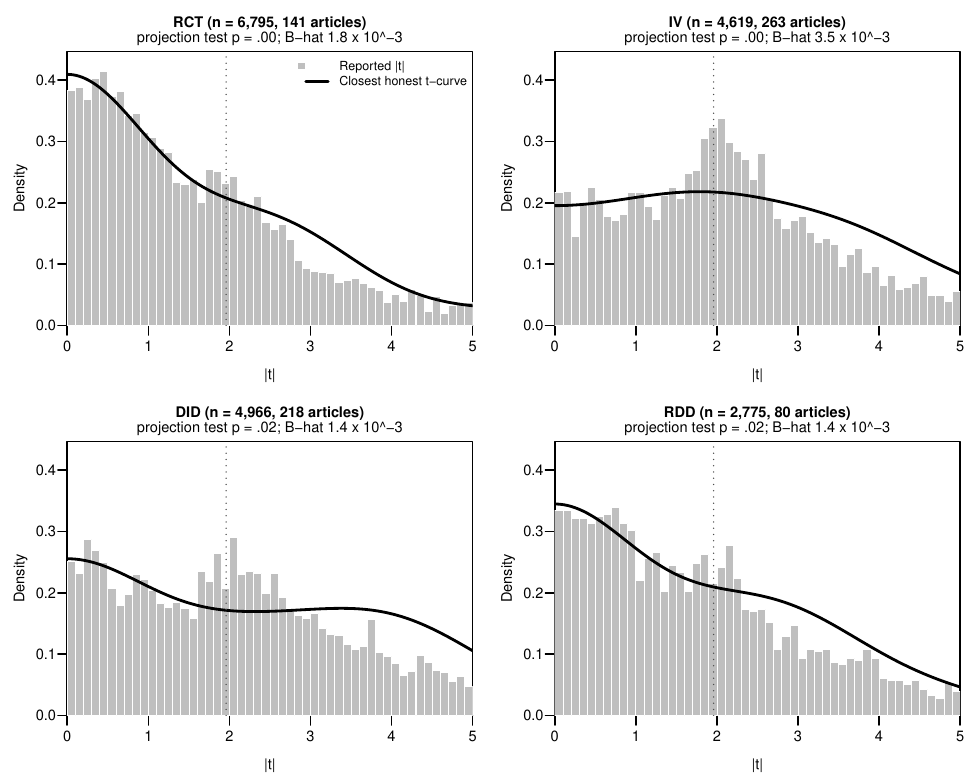}
    \caption{Where does the projection test rejection come from? Histogram of $T$ vs closest honest $t$-curve when derounding using $s_{min}+\text{noise}$. Recall that the closest honest $t$-curve is in the shifted and weighted $||\cdot||_T$-norm so it prioritizes fit near $1.96$.}
    \label{fig:Figure_tcurve_shape_smin_noise_keepmu0_drop_se_mi}
\end{figure}

In contrast, the six tests implemented by \cite{Elliott} and the $p$-curve Stouffer test of \cite{Simonsohn,Simonsohn2015_stouffer} do not find evidence of selective reporting after uniform derounding.  \cite{elliott2024powertestsdetectingphacking} found a similar result in their Table 3 and write that ``De-rounding
drastically changes the results and empirical conclusions.'' An advantage of the projection test is that it does not appear to share this sensitivity. New tests by \cite{kudrinjmp24} also reject the null on a different subsample of this dataset under a different rounding approach.

Beyond power, Table \ref{tab:rounding_ekw_keepmu0_drop_se_mi_cochrane} shows that the projection test is robust to rounding and derounding. Robustness takes two forms: stability and certification. First, the projection test is the only test that rejects across all four derounding approaches for all four research designs. Second, the projection test is the only test which can certify that rounding/derounding was not responsible for its rejections.  Following the method in Section \ref{sec:der_procedure}, notice that $10^{3}\widehat{B}$ is never below its allowance so no projection test rejection is overturned. 

In the Cochrane data, Table \ref{tab:rounding_ekw_keepmu0_drop_se_mi_cochrane} shows no robust evidence of detectable selective reporting. This is an important finding for two reasons. First, this non-rejection demonstrates that the projection test is not so sensitive to tiny real-world deviations from normality that it vacuously rejects everything in large meta-samples. Second, having in hand a large real-world dataset where there is no evidence of detectable selective reporting will allow us to construct a simulation exercise in the next section.

As a robustness check, Appendix Table \ref{tab:rounding_appendix_keepmu0_drop_se_mi_cochrane} shows that changing the projection test weights to $J_0=30$ does not change the results but removing the shift described in Remark \ref{rem:shifting} removes rejections except for RCTs and IVs. This is unsurprising because removing the shift directs the projection test's attention towards $T=0$ where we expect little selective reporting.

\begin{center}
\begin{table}[htbp]
\centering
\caption{Tests for selective reporting under four treatments of rounding, in economics (keeping coefficients reported as zero) and medicine}
\label{tab:rounding_ekw_keepmu0_drop_se_mi_cochrane}
\small
\resizebox{\ifdim\width>\linewidth\linewidth\else\width\fi}{!}{%
\begin{tabular}{lccccccccc}
\toprule
 & \multicolumn{2}{c}{Projection test} & \multicolumn{6}{c}{EKW} & \multicolumn{1}{c}{$p$-curve} \\
\cmidrule(lr){2-3} \cmidrule(lr){4-9} \cmidrule(lr){10-10}
Treatment & $p$ & $10^3\widehat{B}$ & Disc. & Binom. & CS1 & CS2B & LCM & Fisher & Stouffer \\
\midrule
\multicolumn{10}{l}{\textit{RCT}: 6,795 $t$-statistics in 141 articles;  $10^3\times $ rounding allowance $<0.1$} \\
None & \textbf{.00} &  & \textbf{.04} & \textbf{.01} & .47 & .78 & .64 & 1.00 & 1.00 \\
Unif. & \textbf{.00} &  & .55 & .89 & .52 & .46 & 1.00 & 1.00 & 1.00 \\
$s_{\min}$ & \textbf{.00} & 1.8 & .62 & .07 & .80 & .86 & .80 & 1.00 & 1.00 \\
$s_{\min}$+noise & \textbf{.00} & 1.8 & .43 & .94 & .49 & .52 & 1.00 & 1.00 & 1.00 \\
\addlinespace
\multicolumn{10}{l}{\textit{IV}: 4,619 $t$-statistics in 263 articles;  $10^3\times $ rounding allowance $<0.1$} \\
None & \textbf{.00} &  & .88 & .50 & .61 & .05 & 1.00 & 1.00 & 1.00 \\
Unif. & \textbf{.00} &  & .62 & .78 & .56 & .38 & 1.00 & 1.00 & 1.00 \\
$s_{\min}$ & \textbf{.00} & 3.5 & .23 & .20 & .96 & .45 & 1.00 & 1.00 & 1.00 \\
$s_{\min}$+noise & \textbf{.00} & 3.5 & .53 & .67 & .83 & .80 & 1.00 & 1.00 & 1.00 \\
\addlinespace
\multicolumn{10}{l}{\textit{DID}: 4,966 $t$-statistics in 218 articles;  $10^3\times $ rounding allowance $<0.1$} \\
None & \textbf{.01} &  & \textbf{.05} & \textbf{.03} & .78 & \textbf{.04} & .33 & 1.00 & 1.00 \\
Unif. & \textbf{.01} &  & .50 & .95 & .64 & .45 & 1.00 & 1.00 & 1.00 \\
$s_{\min}$ & \textbf{.01} & 1.4 & \textbf{.03} & .08 & .46 & \textbf{.03} & .93 & 1.00 & 1.00 \\
$s_{\min}$+noise & \textbf{.02} & 1.4 & .40 & .94 & .64 & .39 & 1.00 & 1.00 & 1.00 \\
\addlinespace
\multicolumn{10}{l}{\textit{RDD}: 2,775 $t$-statistics in 80 articles;  $10^3\times $ rounding allowance $<0.1$} \\
None & \textbf{.03} &  & \textbf{.03} & \textbf{.01} & .68 & .47 & .99 & 1.00 & 1.00 \\
Unif. & \textbf{.04} &  & .51 & .80 & .71 & .63 & 1.00 & 1.00 & 1.00 \\
$s_{\min}$ & \textbf{.02} & 1.5 & .99 & .35 & .89 & .59 & 1.00 & 1.00 & 1.00 \\
$s_{\min}$+noise & \textbf{.02} & 1.4 & .57 & .94 & .72 & .52 & 1.00 & 1.00 & 1.00 \\
\midrule
\multicolumn{10}{l}{\textit{Cochrane, all trials}: 17,680 outcomes in 5,310 trials; $10^3\times $ rounding allowance $<0.1$} \\
None & .99 &  & .60 & .67 & .28 & .32 & 1.00 & 1.00 & 1.00 \\
Unif. & .99 &  & .49 & .76 & .50 & .51 & 1.00 & 1.00 & 1.00 \\
$s_{\min}$ & .99 & 0 & .53 & .67 & .22 & .37 & 1.00 & 1.00 & 1.00 \\
$s_{\min}$+noise & .99 & 0 & .53 & .73 & .51 & .51 & 1.00 & 1.00 & 1.00 \\
\bottomrule
\end{tabular}}

\medskip
\begin{minipage}{\linewidth}
\footnotesize
\textit{Notes:} $p$-values. Brodeur et al.\ (2020): $t$-statistics by research design. Cochrane: $z$-statistics of continuous efficacy outcomes of randomized trials with at least 100 participants.
Rows: \textit{None}, printed statistics; \textit{Unif.}, every printed number drawn uniformly within its rounding interval; $s_{\min}$, the smallest standard error consistent with the printed numbers (Section 5.2); $s_{\min}$+noise, the same plus uniform noise at the estimate's last digit. \textit{Unif.} and $s_{\min}$+noise average over 200 draws.
Statistics with $\kappa>0.5$ are left out (Section 5.2).
Projection test: $J_0=20$, $J=40$, centered at 1.96. $10^3\widehat{B}$: a lower 95\% confidence bound on the distance to the nearest honest $t$-curve; the rounding allowance (first line of each block), $0.13\,\overline{\kappa^6}$, bounds what rounding can create (Propositions 3 and 4).
EKW, as in Elliott et al.\ (2022): density discontinuity at .05; binomial on $[.04,.05]$; on $(0,.15]$, Cox--Shi monotonicity (CS1) and 2-monotonicity with bounds (CS2B), and Least Concave Majorant (LCM) and Fisher tests of a concave CDF.
Stouffer: the $p$-curve test for left skew of Simonsohn et al.\ (2014).
The projection test, CS1 and CS2B cluster by article or trial; the others do not.
Bold: $p<.05$; in the $s_{\min}$ rows, the projection test only when $\widehat{B}$ exceeds the allowance.
Projection tests with $J_0=30$ and without the shift: Table~\ref{tab:rounding_appendix_keepmu0_drop_se_mi_cochrane}.
\end{minipage}
\end{table}

\end{center}

\section{Simulations from Resampling Real Medical Trials}\label{sec:sims}

We use simulations to compare the size and power of the projection test to the tests of \cite{Elliott} and \cite{Simonsohn,Simonsohn2015_stouffer} and their sensitivity to derounding procedure.  To make the simulation as realistic as possible, we use only real $t$-scores from the Cochrane medical trials. For each simulation repetition we redraw 1,500 out of 3282 randomized trials with at least 2 continuous efficacy outcomes of at least 100 participants without replacement.  Then fraction $(1-\tau) \in [0,1]$ of trials are randomly chosen to be ``honest" and report only the $t$-score corresponding to their primary outcome. The other $\tau$ trials will $p$-hack in one of two ways. In one DGP, $p$-hackers report only the most significant outcome. In a second DGP, the $p$-hackers report their primary outcome if it is significant and the more significant of it and one other random outcome otherwise. This simulation therefore inherits the actual distribution of effects and signal-to-noise ratios, within-trial dependence, and rounding of the real medical trials. It simulates only the outcome reporting and subsample. 

Figure \ref{fig:cochrane_switching_power_desktop_combined} shows the results. It compares the power and size across eight tests and four derounding procedures. The left column shows the ``report the most significant outcome $p$-hacking" and the right column shows the threshold outcome switching DGP. The x-axis of each graph is the fraction $\tau$ of $p$-hackers and the y-axis is the rejection rate. The first row is the projection test of this paper  and the last row takes the smallest $p$-value from the tests of \cite{Elliott} plus the $p$-curve test of \cite{Simonsohn,Simonsohn2015_stouffer}. Each graph has four power curves, each corresponding to a different derounding procedure: no derounding, adding uniform noise to numerator and denominator, replacing the standard error with $s_{min}$ (its smallest possible unrounded value), and replacing the standard error with $s_{min}$ plus adding uniform noise to the numerator of the $t$-score. We ran 500 simulation repetitions with 200 bootstrap repetitions each.

Figure \ref{fig:cochrane_switching_power_desktop_combined} shows that, out of these eight tests, the projection test is the only one that simultaneously controls size (at $\tau=0$) while also having meaningful power to detect the two kinds of $p$-hacking.  The projection test is again insensitive to derounding. The six tests of \cite{Elliott} and  the $p$-curve Stouffer test of \cite{Simonsohn,Simonsohn2015_stouffer} do not appear to have meaningful ability to detect either kind of $p$-hacking in this context. The bottom panel shows that even taking the smallest $p$-value over all seven tests does not result in meaningful rejection rates. Appendix Table \ref{tab:cochrane_power_all_tests} provides test-by-test results. {It should be noted that \cite{kudrinjmp24} proposes three new tests that can have high power compared to the \cite{Elliott} tests in his simulations. Future work should compare them to the projection test.}

\begin{figure}[h!]
    \centering
    \includegraphics[width=1.00\linewidth]{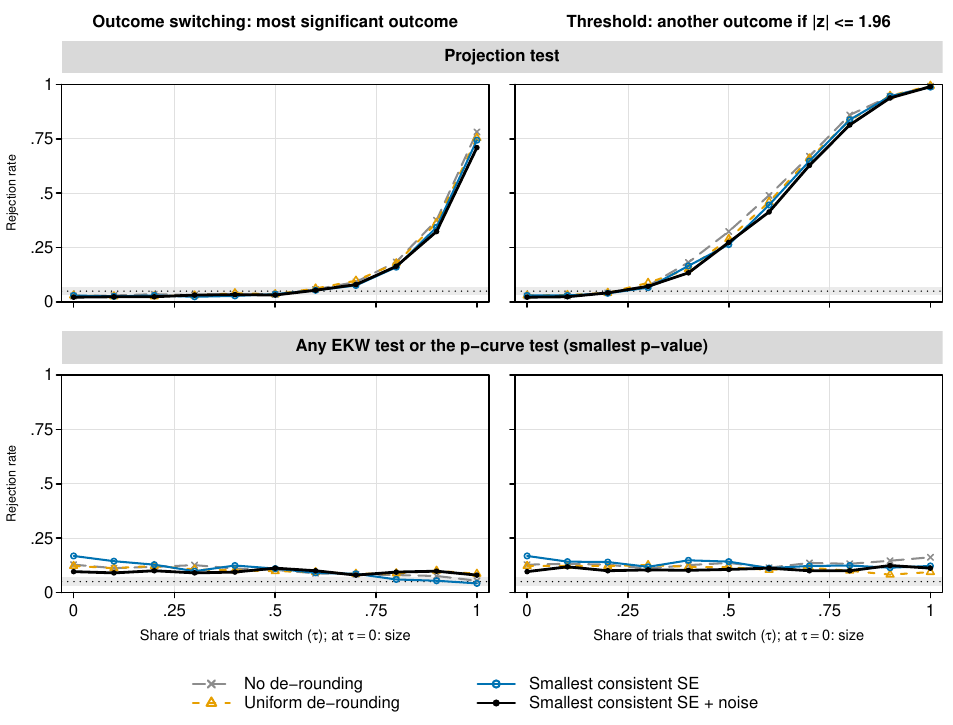}
    \caption{Simulation from real medical trials. The $x$-axis is the share $\tau$ of $p$-hackers.  Columns are $p$-hacking strategies and rows are detection methods. The bottom row is the minimum $p$-value over all the tests of \cite{Elliott} and the $p$-curve Stouffer test from \cite{Simonsohn,Simonsohn2015_stouffer}.  Appendix Table \ref{tab:cochrane_power_all_tests} provides test-by-test results. Data: \cite{Schwab2020}.}
    \label{fig:cochrane_switching_power_desktop_combined}
\end{figure}
\clearpage

\section{Conclusion}

This paper shows that an innocent visual appearance of the $t$-curve can be misleading. Selective reporting can be undetectable. Even when it is detectable, it can manifest in the $t$-curve without the usual signatures and no test has local power when the true effects are sufficiently heterogeneous. We propose a new test that is both consistent against every detectable alternative and also robust to basic distortions coming from (de)rounding. There are two takeaway messages for practice: Meta-analysts should consider using our projection test and they should also avoid interpreting a smooth $t$-curve as evidence of honesty.

\bibliographystyle{chicago}
\bibliography{biblio}
\appendix

\section{Proofs}\label{sec:proofs}

\subsection{Proof of Proposition \ref{prop:self_fulfilling_selection}}\label{proof:prop:self_fulfilling_selection}

\begin{proof}
Fix a bounded measurable $q:\mathbb{R}\times\mathcal{Y}\to\mathbb{R}$. Since $q(T,Y)$ is $(T,Y)$-measurable and $T=H+Z$, the law of iterated expectations and the hypothesis give:
\begin{align*}
&\mathbb{E}\big[R\,q(T,Y)\big]\\
&\qquad=\mathbb{E}\big[\mathbb{P}(R\mid T,Y)\,q(T,Y)\big]\\
&\qquad=\mathbb{E}\big[\mathbb{E}[g(H,Y)\mid T,Y]\,q(T,Y)\big]\\
&\qquad=\mathbb{E}\big[g(H,Y)\,q(H+Z,Y)\big]
\end{align*}

Let  $\mathbb{P}_Z$ and $\mathbb{P}_{(H,Y)}$ denote the distributions of $Z$ and $(H,Y)$. Since $Z\independent(H,Y)$ by hypothesis, the joint distribution of $(H,Y,Z)$ is the product of $\mathbb{P}_{(H,Y)}$ and $\mathbb{P}_Z$. Then by Fubini's Theorem and the fact that $g(h,y)$ does not depend on $z$ :
\begin{align*}
    \mathbb{E}\big[g(H,Y)\,q(H+Z,Y)\big]
&=\int\!\!\int g(h,y)\,q(h+z,y)\,d\mathbb{P}_Z(z)\,d\mathbb{P}_{(H,Y)}(h,y)\\
&= \int\left(\!\!\int \,q(h+z,y)\,d\mathbb{P}_Z(z)\,\right)g(h,y)d\mathbb{P}_{(H,Y)}(h,y)
\end{align*}

 Then recall the definition $\Gamma'(B)=\mathbb{E}\big[g(H,Y)\mathbf{1}\{(H,Y)\in B\}\big]/\mathbb{E}[g(H,Y)]$, which says that $\Gamma'$ has density $g(h,y)/\mathbb{E}[g(H,Y)]$ with respect to $\mathbb{P}_{(H,Y)}$. So:
\begin{align*}
  \int\left(\!\!\int \,q(h+z,y)\,d\mathbb{P}_Z(z)\,\right)g(h,y)d\mathbb{P}_{(H,Y)}(h,y)
&=\mathbb{E}[g(H,Y)]\int\left(\!\!\int \,q(h+z,y)\,d\mathbb{P}_Z(z)\,\right)d\Gamma'(h,y)
\end{align*}

The joint distribution of $(H',Y',Z')$ is the product of $\Gamma'$ and $\mathbb{P}_Z$, because $(H',Y')\sim\Gamma'$, $Z'\stackrel{d}{=}Z$ and $Z'\independent(H',Y')$. So:
\begin{align*}
    \mathbb{E}[g(H,Y)]\int\left(\!\!\int \,q(h+z,y)\,d\mathbb{P}_Z(z)\,\right)d\Gamma'(h,y)
&=\mathbb{E}[g(H,Y)]\,\mathbb{E}\big[q(H'+Z',Y')\big]
\end{align*}

So we have proven that: $$\mathbb{E}\big[R\,q(T,Y)\big]=\mathbb{E}[g(H,Y)]\,\mathbb{E}\big[q(H'+Z',Y')\big]$$

Taking $q\equiv1$ gives $\mathbb{P}(R)=\mathbb{E}[g(H,Y)]$.  so $\mathbb{E}[q(T,Y)\mid R]=\frac{\mathbb{E}\big[R\,q(T,Y)\big]}{\mathbb{E}[g(H,Y)]}=\mathbb{E}[q(H'+Z',Y')]$ for every bounded measurable $q$. This proves $(T,Y)\mid R\stackrel{d}{=}(H'+Z',Y')$. Taking $q$ to depend only on its first argument gives $T\mid R\stackrel{d}{=}H'+Z'$, where $H'\sim\Pi'$ because $\Pi'(A)=\Gamma'(A\times\mathcal{Y})$.

\end{proof}

\subsection{Proof of Theorem \ref{thm:testable_implication}}\label{proof:thm:testable_implication} 

{\noindent \bf   $\implies$ Direction} 

First we prove the $\implies$ direction. If $\mathbf{H}_0$ is true, then $f_{T\mid R} = \int_{-\infty}^\infty \varphi(t-h)d\Pi_0(h)$ for some (possibly non-discrete) probability distribution $\Pi_0$. So $T \mid R$ is continuous and must have a PDF of bounded height. Regardless of $\Pi_0$, we can always construct a sequence of discrete distributions $\Pi_n$ that converge weakly to  $\Pi_0$. Notice that the function $\varphi(t-h)$ is bounded. By the Portmanteau Theorem:
\begin{equation}\label{eq:moment_convergence_weak}
    \lim_{n\to \infty}\mathbb{E}_{\Pi_n}[\varphi(t-h)]  = \mathbb{E}_{\Pi_0}[\varphi(t-h)] \: \forall\: t\in \mathbb{R}
\end{equation}

\noindent So the function $\varphi(t)\left(\mathbb{E}_{\Pi_0}[\varphi(t-h)] - \mathbb{E}_{\Pi_n}[\varphi(t-h)]\right)^2$ converges pointwise to zero:
\begin{equation}
    \lim_{n\to \infty} \varphi(t)\left(\mathbb{E}_{\Pi_0}[\varphi(t-h)] - \mathbb{E}_{\Pi_n}[\varphi(t-h)]\right)^2 = 0
\end{equation}

Since $\left|\varphi(t)\left(\mathbb{E}_{\Pi_0}[\varphi(t-h)] - \mathbb{E}_{\Pi_n}[\varphi(t-h)]\right)^2\right| \leq \varphi(t)$ it is dominated by an integrable function. By the Dominated Convergence Theorem:
 \begin{align*}
 \lim_{n\to \infty}\left|\left|f_{T\mid R}- \int_{-\infty}^\infty \varphi(t-h)d\Pi_n(h) \right|\right|_T^2 &=\lim_{n\to \infty}\int_{-\infty}^\infty \varphi(t)\left(\mathbb{E}_{\Pi_0}[\varphi(t-h)] - \mathbb{E}_{\Pi_n}[\varphi(t-h)]\right)^2dt \\
 &= \int_{-\infty}^\infty \lim_{n\to \infty}\varphi(t)\left(\mathbb{E}_{\Pi_0}[\varphi(t-h)] - \mathbb{E}_{\Pi_n}[\varphi(t-h)]\right)^2dt\\
&= 0
 \end{align*}

 So there is a sequence of $(\Pi_n)_{n=1}^\infty \subseteq \mathcal{D}$  such that $ \left|\left|f_{T\mid R}- \int_{-\infty}^\infty \varphi(t-h)d\Pi_n(h) \right|\right|_T\to 0$. Therefore the infimum over $\mathcal{D}$ must be zero:
 \begin{align*}
       \inf_{\Pi \in \mathcal{D} } \left|\left|f_{T\mid R}- \int_{-\infty}^\infty \varphi(t-h)d\Pi(h) \right|\right|_T =0
 \end{align*}

\vskip 0.1in
 {\noindent\bf   $\impliedby$ Direction}   

By Assumption \ref{assum:cts_TR},  the PDF $f_{T \mid R}$ exists. Suppose that $ \inf_{\Pi \in \mathcal{D} } \left|\left|f_{T\mid R}- \int_{-\infty}^\infty \varphi(t-h)d\Pi(h) \right|\right|_T =0$. Then, there is a sequence $\Pi_n \subset \mathcal{D}$ such that $\left|\left|f_{T\mid R}- \int_{-\infty}^\infty \varphi(t-h)d\Pi_n(h) \right|\right|_T \to 0$. So:
\begin{align*}
    \lim_{n\to \infty}\left|\left|f_{T\mid R}- \int_{-\infty}^\infty \varphi(t-h)d\Pi_n(h) \right|\right|_T^2 &=\lim_{n\to \infty}\int_{-\infty}^\infty \varphi(t)\left(f_{T\mid R}(t) - \mathbb{E}_{\Pi_n}[\varphi(t-h)]\right)^2dt  = 0
\end{align*}

Therefore $\lim_{n\to \infty}\int_{-\infty}^\infty \ell_n(t) dt\to 0$ where $\ell_n(t)\equiv \varphi(t)\left(f_{T\mid R}(t) - \mathbb{E}_{\Pi_n}[\varphi(t-h)]\right)^2$. By Proposition 1.1.9 of \cite{Grafakos2010ClassicalFA}, $\ell_n \to 0$ in measure. By Theorem 1.1.11 of  \cite{Grafakos2010ClassicalFA}, there is a subsequence $k_n$ such that $\ell_{k_n} \to 0$ almost everywhere. Since $\varphi(t)>0$ this implies that $f_{T\mid R}(t) - \mathbb{E}_{\Pi_{k_n}}[\varphi(t-h)] \to 0$ for almost all $t$. In other words, $ \int_{-\infty}^\infty \varphi(t-h)d\Pi_{k_n}(h)\to f_{T\mid R}(t)$ pointwise for almost every $t$. 

Since both are PDFs, $ \int_{-\infty}^\infty\int_{-\infty}^\infty \varphi(t-h)d\Pi_{k_n}(h)dt= \int_{-\infty}^\infty f_{T\mid R}(t)dt=1$. So by Scheff\'e's Lemma, $$ \int_{-\infty}^{\infty}\left|\mathbb{E}_{\Pi_{k_n}}[\varphi(t-h)]- f_{T\mid R}(t)\right|dt\to 0$$

Let $\phi_{\Pi_{k_n}+Z}$ be the characteristic functions of $\mathbb{E}_{\Pi_{k_n}}[\varphi(t-h)]$ and let $\phi_{T\mid R}$ be the characteristic function of $f_{T\mid R}$. Then, for any $\zeta \in \mathbb{R}$:
\begin{align*}
     \left|\phi_{\Pi_{k_n}+Z}(\zeta) - \phi_{T\mid R}(\zeta)\right| &\leq \left|\int_{-\infty}^\infty e^{i\zeta t}\left(\mathbb{E}_{\Pi_{k_n}}[\varphi(t-h)]- f_{T\mid R}(t)\right) dt \right|\\
     &\leq  \int_{-\infty}^{\infty}\left|\mathbb{E}_{\Pi_{k_n}}[\varphi(t-h)]- f_{T\mid R}(t)\right|dt\to 0
\end{align*}

So $\phi_{\Pi_{k_n}+Z}(\zeta) \to \phi_{T\mid R}(\zeta)$ pointwise. By the Convolution Theorem, $\phi_{\Pi_{k_n}+Z}(\zeta)  = \phi_{\Pi_{k_n}}(\zeta)\exp(-\zeta^2/2)$. So:
\begin{align*}
     \phi_{\Pi_{k_n}}(\zeta)\exp(-\zeta^2/2) &\to  \phi_{T\mid R}(\zeta)\\
      \phi_{\Pi_{k_n}}(\zeta) &\to  \phi_{T\mid R}(\zeta)\exp(\zeta^2/2)
\end{align*}

Notice that $\exp(\zeta^2/2)$ is continuous and equals one at $\zeta=0$. Since $ \phi_{\Pi_{k_n}}(\zeta) $ are valid characteristic functions and $\phi_{T\mid R}(\zeta)\exp(\zeta^2/2)$ is continuous at zero, by the L\'evy Continuity Theorem, $\phi_{T\mid R}(\zeta)\exp(\zeta^2/2)$ must also be a valid characteristic function of a random variable with distribution $\Pi$. So $\phi_{T\mid R}(\zeta) = \phi_{\Pi}(\zeta)\exp(-\zeta^2/2)$. This is a convolution of the distribution $\Pi$ with an independent standard normal. So $\mathbf{H}_0$ is true: $T \mid R \stackrel{d}{=} H'+Z$ where $H'\independent Z$ and $H'\sim \Pi$ and $Z\sim N(0,1)$. 

\subsection{Proof of Lemma \ref{lem:ynorm_infsum}}\label{proof:lem:ynorm_infsum}


We have assumed that  $f_{T \mid R} \in \mathcal{L}_T$. Notice that $\sup_{t\in\mathbb{R}}|\int_{-\infty}^\infty \varphi(t-h)d\Pi(h)|\leq \sup_{x\in\mathbb{R}}|\varphi(x)|\leq \frac{1}{\sqrt{2\pi}}$. So $\int_{-\infty}^\infty \varphi(t-h)d\Pi(h)$ is a PDF in $\mathcal{L}_T$. Since  $\mathcal{L}_T$ is a Hilbert space their difference is a member of $\mathcal{L}_T$ and the $||\cdot||_T$ norm of the difference is therefore finite. Since $\psi_j$ form an orthonormal basis of $\mathcal{L}_T$, by Parseval's Formula (Theorem 2.8 of \cite{CARRASCO2Handbook}):
\begin{align*}
      \left|\left|f_{T\mid R}- \int_{-\infty}^\infty \varphi(t-h)d\Pi(h) \right|\right|_T &= \sqrt{  \sum_{j=0}^\infty    \left(\left\langle f_{T\mid R}- \int_{-\infty}^\infty \varphi(t-h)d\Pi(h),\psi_j\right\rangle_T\right)^2 } \\
      &= \sqrt{  \sum_{j=0}^\infty    \left(\left\langle f_{T\mid R},\psi_j\right\rangle_T- \left\langle\int_{-\infty}^\infty \varphi(t-h)d\Pi(h),\psi_j\right\rangle_T\right)^2 } 
\end{align*}

Since $\int_{-\infty}^\infty \varphi(t-h)d\Pi(h)$ is the density of $H'+Z$ and integrals over densities are expectations:
\begin{align*}
     &\sqrt{  \sum_{j=0}^\infty    \left(\left\langle f_{T\mid R},\psi_j\right\rangle_T- \left\langle\int_{-\infty}^\infty \varphi(t-h)d\Pi(h),\psi_j\right\rangle_T\right)^2 }\\
      &=  \sqrt{  \sum_{j=0}^\infty    \left(\int_{-\infty}^\infty f_{T\mid R}(t)\psi_j(t)\varphi(t)dt- \int_{-\infty}^\infty \left[\int_{-\infty}^\infty \varphi(t-h)d\Pi(h)\right]\psi_j(t)\varphi(t)dt \right)^2 } \\
     &= \sqrt{  \sum_{j=0}^\infty    \left(\mathbb{E}\left[\psi_j(T)\varphi(T)\mid R \right] - \mathbb{E}\left[\psi_j(H'+Z)\varphi(H'+Z) \right]\right)^2 } \\
     &=\sqrt{  \sum_{j=0}^\infty    \left(\theta_{0,j} -c_j(\Pi)\right)^2 }
\end{align*}

 This completes the proof.


\subsection{Proof of Lemma \ref{lem:cj_Pi}}\label{proof:lem:cj_Pi} 

This proof is based on Example 1 of \cite{CarrascoPaper} and was repurposed by \cite{faridani2025testingunderpoweredliteratures}. It is not new and included here only for completeness. The proof proceeds in two steps. First we derive an important property of the adjoint of the convolution operator. Second, we use Portmanteau to show the result.  

{\bf\noindent Step 1: The Adjoint}

Define $\varphi_2$ as the PDF of the Gaussian with variance 2. We will need to define the following inner product $  \langle \cdot , \cdot \rangle_{H} $ and its Hilbert Space $ \mathcal{L}_{H}$:
\begin{align}
     \langle \phi_1, \phi_2 \rangle_{H} &\equiv  \int_{-\infty}^\infty \phi_1(x)\phi_2(x)\varphi_2(x)dx\\
     ||\phi||_H &\equiv  \sqrt{\langle \phi, \phi \rangle_{H}}\\
    \mathcal{L}_{H} &\equiv \left\{\phi(x)\text{ such that } ||\phi||_H < \infty \right\} 
\end{align}

Next define the convolution operator $K \: : \: \mathcal{L}_{H} \to \mathcal{L}_{T} $ which convolves a function $g(h) \in \mathcal{L}_{H}$ with the standard Gaussian PDF:
\begin{align}
    (Kg)[t] &\equiv \int_{-\infty}^\infty g(h)\varphi(t-h)dh
\end{align}

We have defined these two spaces $\mathcal{L}_T,\mathcal{L}_H$ so that $K$ is Hilbert-Schmidt. By Example 1 of \cite{CarrascoPaper}, the Singular Value Decomposition of $K$ is:
\begin{align}\label{eq:SVD}
     Kg = \sum_{j=0}^\infty \frac{1}{2^{j/2}}\langle g,\chi_j \rangle_H \psi_j
\end{align}

By the orthonormality of the Hermite polynomials and Equation (\ref{eq:SVD}):
\begin{align}
   \langle K g, \psi_j \rangle_T &= \frac{1}{2^{j/2}}\langle g,\chi_j \rangle_H ,\qquad \forall g \in \mathcal{L}_H
\end{align}

Define $K^*$ as the adjoint of $K$ with respect to the inner products $  \langle \cdot , \cdot \rangle_{H},  \langle \cdot , \cdot \rangle_{T}  $. Therefore:
\begin{align}
   \langle  g, K^*\psi_j \rangle_H &= \langle K g, \psi_j \rangle_T = \left\langle g, \frac{1}{2^{j/2}}\chi_j \right\rangle_H
\end{align}

Since this is true for all $g \in \mathcal{L}_H$ and by the uniqueness of the Riesz representor:
\begin{align}\label{eq:adjoint}
    K^*\psi_j  &= \frac{1}{2^{j/2}}\chi_j
\end{align}

{\bf\noindent Step 2: Portmanteau Argument}

 We will start by proving the result for smooth $\Pi$  and then extend to all probability distributions via weak convergence. For any probability distribution $\Pi$ we can define a sequence of continuous distributions $H_n\sim\Pi_n\equiv \Pi * N\left(0,\tau_n^2\right)$  where $\tau_n\to 0$ and $*$ denotes convolution. So $\Pi_n$ has PDF $\pi_n$ each of finite height and $\Pi_n\to_w \Pi$ (where $\to_w$ denotes weak convergence). Since each $\pi_n$ has finite height, we have $\pi_n\in\mathcal L_H$. 
 \begin{align*}
   c_j(\Pi_n)&=  \mathbb{E}_{\Pi_n}\left[ \varphi(H_n+Z)\psi_j(H_n+Z) \right] &\text{def. of $c_j$}\\
   &= \int_{-\infty}^\infty \varphi(t)\psi_j(t)[K\pi_n](t)dt & \text{$K\pi_n$ is PDF of $H_n+Z$}\\
   &= \langle K\pi_n,\psi_j \rangle_T &\text{def. of $\langle \cdot, \cdot \rangle_T$}\\
     &= \langle \pi_n,K^*\psi_j \rangle_H  &\text{def. of the Adjoint}\\
     &= \frac{1}{2^{j/2}} \langle \pi_n, \chi_j\rangle_H  &\text{Equation (\ref{eq:adjoint})}\\
     &=  \frac{1}{2^{j/2}} \int_{-\infty}^\infty \varphi_2(h)\chi_j(h)\pi_n(h)dh & \text{def. of $\langle \cdot, \cdot \rangle_H$}\\
     &=   \frac{1}{2^{j/2}} \mathbb{E}_{\Pi_n}\left[ \varphi_2(H)\chi_j(H) \right]  &\text{$\pi_n$ is PDF of $\Pi_n$}
 \end{align*}

By Lemma \ref{lem:bound_coeffs}, the functions $ \varphi(t)\psi_j(t)$ and $ \varphi_2(h)\chi_j(h)$ are uniformly bounded. They are continuous since they are the products of the Gaussian PDF and polynomials. So by Portmanteau:
\begin{align*}
     \mathbb{E}_{\Pi_n}\left[ \varphi(H_n+Z)\psi_j(H_n+Z) \right] &\to  \mathbb{E}_{\Pi}\left[ \varphi(H+Z)\psi_j(H+Z) \right]\\
     \mathbb{E}_{\Pi_n}\left[ \varphi_2(H_n)\chi_j(H_n) \right] &\to \mathbb{E}_{\Pi}\left[ \varphi_2(H)\chi_j(H) \right]
\end{align*}

Since the sequences are equal, their limits must also be equal. So for all probability distributions $\Pi$:
\begin{align*}
    c_j(\Pi) =  \mathbb{E}_{\Pi}\left[ \varphi(H+Z)\psi_j(H+Z) \right]=   \frac{1}{2^{j/2}}\mathbb{E}_{\Pi}\left[ \varphi_2(H)\chi_j(H) \right]
\end{align*}

\subsection{Proof of Lemma \ref{lem:H0_dw}}\label{proof:lem:H0_dw}

Let $d(\theta)\equiv \inf_{\Pi \in \mathcal{D}} \sqrt{ \sum_{j=0}^\infty (\theta_j-c_j(\Pi))^2}$ which is (\ref{eq:Dj_definition}) with every $w_j=1$. So $\mathbf{H}_0\iff d(\theta_0)=0$
by Theorem~\ref{thm:testable_implication} and Lemma~\ref{lem:ynorm_infsum}. Since $w_j\le1$, $d_w\le d$, we only need to show that
$d_w(\theta_0)=0 \implies d(\theta_0)=0$. Take $\Pi_k\in\mathcal{D}$ with
$\|W(\theta_0-c(\Pi_k))\|_{\ell_2}\to0$. By Lemmas \ref{lem:bound_coeffs} and \ref{lem:cj_Pi},
$|c_j(\Pi)|\leq  2^{-j/2}$ for every $\Pi$. Since
$\sum_j2^{-j}<\infty$, the set $\{x\in\ell_2:|x_j|\le 2^{-j/2}\ \forall j\}$ is compact,
so along a subsequence $c(\Pi_k)\to x$ in $\ell_2$. Then $Wc(\Pi_k)\to Wx$, so
$Wx=W\theta_0$, and $x=\theta_0$ because every $w_j>0$. Therefore
$\|\theta_0-c(\Pi_k)\|_{\ell_2}\to0$ along the subsequence, so
$d(\theta_0)=0$.

\subsection{Proof of Lemma \ref{lem:compact_Lambda}}\label{proof:lem:compact_Lambda}

    Let $S \equiv \{W\mathbf{b}_x \: : \: x \in \mathbb{R}\}\cup \{0\}\subseteq \ell^2$. First we show that $S$ is compact. By Parseval's Theorem, $||b_x||_{\ell^2} = ||\varphi(t-x)||_T = (2\pi \sqrt{3})^{-1/2}e^{-x^2/6}$ and $||b_x-b_y||_{\ell^2} = ||\varphi(t-x)-\varphi(t-y)||_T \leq \varphi(0)|x-y|$. So $x \to W{\mathbf{b}}_x$ is continuous and $\lim_{|x|\to\infty}||\mathbf{b}_{x}||_{\ell^2} \to 0$. Since $w_j \in (0,1]$ for all $j$,  $||W\mathbf{b}_x||_{\ell^2} \leq ||\mathbf{b}_x||_{\ell^2}$. Define the function $g(u)=Wb_{u/(1-|u|)}$ for $|u|<1$ and $g(1)=g(-1)=0$. Since it is the composition of continuous functions, $g$ is continuous. The set $[-1,1]$ is compact. Since any  $x\in \mathbb{R}$ can be written as $u/(1-|u|)$ for some $u$, we have $S = g([-1,1])$. Since it is the continuous image of a compact set, $S$ is compact. Since $\ell_2$ is a Hilbert space, and hence a Fréchet space, we can apply \cite{rudin1991functional}  Theorem 3.20(c) to conclude that $\overline{\mathrm{co}}(S)$ is compact. It is convex because it is the closure of a convex set. Since $\Lambda_w$ is convex and $\lim_{x\to\infty}||\mathbf{b}_{x}||_2 \to 0$ ,  $ \Lambda_w = \overline{\text{conv}}\left(S\right)$ because $\{0\} \in \overline{\{W\mathbf{b}_x \: : \: x \in \mathbb{R}\}}$. So  $\Lambda_w$ is compact and convex.

\subsection{Proof of Lemma \ref{lem:clt_and_bootstrap}}\label{proof:lem:clt_and_bootstrap}

    {\bf Step 1: A Useful Bound} 
    
    Lemma~\ref{lem:bound_coeffs} guarantees that $\sup_{t,j}|\varphi(t)\psi_j(t)|\le(2\pi)^{-1/2}$. By Assumption \ref{assum:articles}, each article reports at most $c_0$ $t$-statistics, so for every article $g$ and every $j$,
\begin{equation}\label{eq:art_bound}
  \left(\sum_{i\in g}\big(\varphi(t_i)\psi_j(t_i)-\theta_{0,j}\big)\right)^2\le\frac{2c_0^2}{\pi}\qquad  \forall_{g \in 1,\cdots m}
\end{equation}
and the same bound holds with $\widehat{\theta}_{nj}$ in place of $\theta_{0,j}$.
    
\noindent {\bf Step 2: CLT}  

Let $S_g \equiv \left(w_j\sum_{i\in g}(\varphi(t_i)\psi_j(t_i)\right)_{j=0}^{\infty}$ and  $X_g\equiv S_g-N_gW\theta_0$. This makes the $j$-th coordinate of $X_g$ $w_j\sum_{i\in g}(\varphi(t_i)\psi_j(t_i)-\theta_{0,j})$. By Assumption~\ref{assum:articles} the $X_g$ are i.i.d.; they have mean zero because $W\theta_0=E[S_g]/E[N_g]$; and by \eqref{eq:art_bound} and Assumption~\ref{assum:weights}, $\|X_g\|_{\ell^2}^2\le(2c_0^2/\pi)\sum_jw_j^2<\infty$. Each coordinate of $X_g$ is a measurable function of article $g$'s data and $\ell^2$ is separable, so $X_g$ is Borel measurable. By Example 1.8.5 of \cite{vandervaart1996weak},
$$\frac{1}{\sqrt m}\sum_{g=1}^mX_g\xrightarrow{d}\mathbb{G}$$
a tight, mean-zero Gaussian element of $\ell^2$ whose $j$-th coordinate $\mathbb{G}_j$ is normal with variance $E[X_{1,j}^2]$, the second moment of the $j$-th coordinate of $X_1$. Because $\widehat{\theta}_{nj}-\theta_{0,j}=n^{-1}\sum_{i=1}^n\big(\varphi(t_i)\psi_j(t_i)-\theta_{0,j}\big)$ and $n=m\overline{N}_m$,
$$
  \sqrt m\,(W\widehat{\theta}_n-W\theta_0)=\frac{1}{\overline{N}_m}\cdot\frac{1}{\sqrt m}\sum_{g=1}^mX_g 
$$
By the strong law of large numbers $\overline{N}_m\to E[N_g]$ almost surely, so by Slutsky and the continuous mapping theorem (Example 1.4.7 and Theorem 1.3.6 of \cite{vandervaart1996weak}), $\sqrt m\,(W\widehat{\theta}_n-W\theta_0)\xrightarrow{d}\mathbb{G}_0\equiv\mathbb{G}/E[N_g]$.

\noindent {\bf Step 3: Bounding the tails}   

Next we show that using only the first $J+1$ elements of either the sample statistic or the bootstrap has vanishing impact as $J$ grows. Define $X_g^*\equiv S_g^*-N_g^*W\widehat{\theta}_n$. By the same argument as the previous step:
$$
  \sqrt m\,(W\widehat{\theta}_n^*-W\widehat{\theta}_n)=\frac{1}{\overline{N}_m^*}\cdot\frac{1}{\sqrt m}\sum_{g=1}^mX_g^* .
$$
Conditional on the sample, $X_1^*,\dots,X_m^*$ are i.i.d.\ draws from the $m$ values $S_h-N_hW\widehat{\theta}_n$, $h=1,\dots,m$. Since $\overline{N}_m^*\geq 1$, \eqref{eq:art_bound} guarantees, for each $J$,
\begin{equation}\label{eq:tail_boot}
  E\Big[m\sum_{j>J}w_j^2\big(\widehat{\theta}^*_{nj}-\widehat{\theta}_{nj}\big)^2\,\Big|\,\text{sample}\Big]
  \le\frac1m\sum_{h=1}^m\sum_{j>J}w_j^2\Big(\sum_{i\in h}\big(\varphi(t_i)\psi_j(t_i)-\widehat{\theta}_{nj}\big)\Big)^2
  \le\frac{2c_0^2}{\pi}\sum_{j>J}w_j^2 .
\end{equation}
The Gaussian limit from the previous step satisfies the same bound: since $E[N_g]\geq 1$, $E[\mathbb{G}_{0,j}^2]\le E[\mathbb{G}_j^2]=E[X_{1,j}^2]\le(2c_0^2/\pi)w_j^2$, so $E\left[\sum_{j>J}\mathbb{G}_{0,j}^2\right]\leq(2c_0^2/\pi)\sum_{j>J}w_j^2$.

{\bf Step 4: Bootstrap}

Fix any $J$ and any $f\in BL_1(\ell^2)$. Because $f$ is assumed to be 1-Lipschitz, setting each coordinate after $J$ to zero changes $f$ by at most the $\ell^2$ norm of those coordinates. Applying this to both $\sqrt m\,W(\widehat{\theta}^*_n-\widehat{\theta}_n)$ and $\mathbb{G}_0$, and using \eqref{eq:tail_boot}, the bound for $\mathbb{G}_0$, and Jensen's inequality,
\begin{multline*}
  \Big|E\Big[f\big(\sqrt m\,(W\widehat{\theta}_n^*-W\widehat{\theta}_n)\big)\,\Big|\,\text{sample}\Big]-Ef(\mathbb{G}_0)\Big|
  \le2\Big(\frac{2c_0^2}{\pi}\sum_{j>J}w_j^2\Big)^{1/2}\\
  +\Big|E\Big[f\big(\sqrt m\,w_0(\widehat{\theta}^*_{n0}-\widehat{\theta}_{n0}),\dots,\sqrt m\,w_J(\widehat{\theta}^*_{nJ}-\widehat{\theta}_{nJ}),0,0,\dots\big)\,\Big|\,\text{sample}\Big]\\
  -Ef\big(\mathbb{G}_{0,0},\dots,\mathbb{G}_{0,J},0,0,\dots\big)\Big| .
\end{multline*}
The last term involves only the first $J+1$ coordinates. Let $(S_{g,J}^,N_{g,J}^,S_{g,J}^*,N_{g,J}^*)$ be the random vectors in $\mathbb{R}^{J+1}$ corresponding to the first $J+1$ coordinates of their respective quantities. By Theorem 23.4 of \cite{vanderVaart1998} and the delta method (Theorem 23.5 in \cite{vanderVaart1998})  and Slutsky, conditional on the sample, the first $J+1$ coordinates of $\sqrt m(\widehat{\theta}^*_n-\widehat{\theta}_n)$ converge in distribution to the same limit as those of $\sqrt m(\widehat{\theta}_n-\theta_0)$.

Multiplying the $j$-th coordinate by $w_j$ preserves this conditional convergence. Hence the first $J+1$ coordinates of $\sqrt m\,W(\widehat{\theta}^*_n-\widehat{\theta}_n)$ converge conditionally in distribution to the same limit as those of $\sqrt m\,W(\widehat{\theta}_n-\theta_0)$, which by the CLT above is the same distribution as $(\mathbb{G}_{0,0},\dots,\mathbb{G}_{0,J})$. For every $f\in BL_1(\ell^2)$, the function $x\mapsto f(x_0,\dots,x_J,0,0,\dots)$ on $\mathbb{R}^{J+1}$ is bounded by one and 1-Lipschitz. Theorem 1.13.1 of \cite{vandervaart1996weak} thus implies that the last term converges to zero in probability, uniformly in $f\in BL_1(\ell^2)$.

So, for every $J$ and $\varepsilon>0$, the probability that the supremum in the statement of the lemma exceeds $2\big((2c_0^2/\pi)\sum_{j>J}w_j^2\big)^{1/2}+\varepsilon$ tends to zero as $m\to\infty$. By Assumption~\ref{assum:weights} this bound can be made arbitrarily small by choosing $J$ large, so the supremum converges to zero in probability. Finally, $\sqrt m\,(W\widehat{\theta}^*_n-W\widehat{\theta}_n)$ is a Borel measurable function of the sample and the bootstrap draws, so Assumption 3 of \cite{FangSantos} holds.

\subsection{Proof of Theorem \ref{thm:validity_and_consistency_ideal}}\label{proof:thm:validity_and_consistency_ideal}

 {\noindent\bf Pointwise Validity:}

By Lemma \ref{lem:compact_Lambda}, $\Lambda_w$ is compact and convex. By Lemma \ref{lem:H0_dw}, $\mathbf{H}_0 \iff d_w(\theta_0)=0 \iff W\theta_0 \in \Lambda_w$.  Lemma \ref{lem:clt_and_bootstrap} guarantees that  the estimator and its bootstrap satisfy Assumptions 2 and 3 of \cite{FangSantos} in $\ell^2$, with rate $\sqrt{m}$ and limit $\mathbb{G}_0$.   Their Proposition 4.1 implies their Assumption 1, so their Theorem~2.1 yields $\sqrt{m}d_w(\widehat{\theta}_n)\xrightarrow{d} \phi'_{\theta_0}(\mathbb{G}_0)$.

By our conditions on $s_n$, their Lemma S.3.8 shows that $\widehat{\phi}'_{n,s_n}$ satisfies their Assumption 4. Their Theorem 3.2 then says that the bootstrap distribution of $\widehat{\phi}'_{n,s_n}(\sqrt{m}(\widehat{\theta}^*_n-\widehat{\theta}_n))$ converges in distribution to the distribution of $\phi'_{\theta_0}(\mathbb{G}_0)$.

Let $c_{1-\alpha}$ denote the true $1-\alpha$ quantile of $\phi'_{\theta_0}(\mathbb{G}_0)$. Because the limit infimum $\eta_m$ is assumed to be bounded below there is some $\eta >0$ such that $\eta_m > \eta$ for sufficiently large $m$. We now show that $\mathbb{P}[{cv}(\alpha)+\eta_m > c_{1-\alpha}+\eta/2]\to 1$. To do this, define the piecewise linear function $g(x)$ which equals one when $x \leq c_{1-\alpha}-\eta/2$, equals zero when $x \geq c_{1-\alpha}-\eta/4$ and is a linear interpolation between those two values. In particular, this means that $\mathbf{1}\{x \leq c_{1-\alpha}-\eta/2\} \leq g(x)\leq \mathbf{1}\{x \leq c_{1-\alpha}-\eta/4\}$. Combining this with Theorem 3.2 of \cite{FangSantos}:
\begin{align*}
  F_n^*(c_{1-\alpha}-\eta/2)\ &\leq E\left[g\left(\widehat{\phi}'_{n,s_n}(\sqrt{m}(\widehat{\theta}^*_n-\widehat{\theta}_n))\right)\big|\{t_i\}_{i=1}^n,\{N_g\}_{g=1}^m\right]\\
  &\to_p E\left[g(\phi'_{\theta_0}(\mathbb{G}_0))\right] \leq \mathbb{P}\left[\phi'_{\theta_0}(\mathbb{G}_0)\leq c_{1-\alpha}-\eta/4\right] < 1-\alpha,
\end{align*}

So $\mathbb{P}[{cv}(\alpha)\leq c_{1-\alpha}-\eta/2]\to 0$ and therefore $\mathbb{P}[{cv}(\alpha)+\eta_m > c_{1-\alpha}+\eta/2]\to 1$. Using this fact and then applying Portmanteau and then the definition of $c_{1-\alpha}$ yields validity:
\begin{align*}
  \limsup_{m\to\infty}\mathbb{P}\left[\sqrt{m}\,d_w(\widehat{\theta}_n)>\overline{cv}(\alpha)\right]
  \ &\leq \limsup_{m\to\infty}\mathbb{P}\left[\sqrt{m}\,d_w(\widehat{\theta}_n)\geq c_{1-\alpha}+\eta/2\right]\\
  \ &\leq \mathbb{P}\left[\phi'_{\theta_0}(\mathbb{G}_0)\ge c_{1-\alpha}+\eta/2\right]\ \leq \alpha ,
\end{align*}

  {\bf\noindent Consistency against fixed alternatives}

  By Lemma \ref{lem:H0_dw} and the fact that all weights $w_j$ are positive, $\mathbf{H}_0\text{ false } \iff d_w(\theta_0)>0$. Since $d_w$ is continuous and $||W\widehat{\theta}_n-W\theta_0||_{\ell^2}\to_p 0$ by Lemma \ref{lem:clt_and_bootstrap} then $d_w(\widehat{\theta}_n) \to_p d_w(\theta_0)>0$. So $\sqrt{m}{d}_w(\widehat{\theta}_n )$ diverges in probability. By Lemma \ref{lem:clt_and_bootstrap}, $\sqrt{m}(W\widehat{\theta}_n^*-W\widehat{\theta}_n)=\mathcal{O}_p(1)$. Moreover, $d_w$ is 1-Lipschitz so $\widehat{\phi}'(h)\leq M||Wh||_{\ell^2}$ for some constant $M>0$. Combining this with the assumption that $\eta_m$ are bounded, $\overline{cv}(\alpha)=\mathcal{O}_p(1)$.  Since the left side diverges while the critical value is bounded in probability, $\lim_{m\to \infty}\mathbb{P}\left[\sqrt{m}{d}_w(\widehat{\theta}_n ) > \overline{cv}(\alpha)\right]=1$.

\subsection{Proof of Lemma \ref{lem:max_grid_approx_error}}\label{proof:lem:max_grid_approx_error} 

Since cutting the sum off at $J$ only decreases it:  $d_{w,J}(\theta)\leq {d}_{w}(\theta)$.

Next we bound the effect of the grid, i.e. $\tilde{d}_{w,J}(\theta)-d_{w,J}(\theta)$. Let $\Pi_y$ denote the point mass at $y\in\mathbb R$, and define
\begin{align*}
\mathbf u_y \equiv \big(w_0c_0(\Pi_y),\dots,w_Jc_J(\Pi_y)\big)\in\mathbb R^{J+1}.
\end{align*}
Here $W\equiv \text{diag}(w_0,\cdots w_J)$. Recall that
\begin{align*}
d_{w,J}(\theta)=\inf_{\Pi\in\mathcal D}\|W\theta-W\mathbf c(\Pi)\|_2,
\qquad
\widetilde d_{w,J}(\theta)=\inf_{\widetilde\Pi\in{\mathcal D}_{\mathcal X}}\|W\theta-W\mathbf c(\widetilde\Pi)\|_2,
\end{align*}
where ${\mathcal D}_{\mathcal X}$ is the set of discrete distributions with support contained in $\mathcal X$.

Define the convex sets of vectors:
\begin{align*}
U \equiv \operatorname{conv}\{\mathbf u_y:y\in\mathbb R\},
\qquad
\widetilde U \equiv \operatorname{conv}\{\mathbf u_x:x\in\mathcal X\}.
\end{align*}
Since every discrete distribution is a convex combination of point masses, we have
\begin{align*}
d_{w,J}(\theta)=\operatorname{dist}(W\theta,U),
\qquad
\widetilde d_{w,J}(\theta)=\operatorname{dist}(W\theta,\widetilde U).
\end{align*}
Because $\widetilde U\subseteq U$,
\begin{align*}
0\leq \widetilde d_{w,J}(\theta)-d_{w,J}(\theta) \leq \sup_{u\in U}\inf_{\widetilde u\in\widetilde U}\|u-\widetilde u\|_2
\end{align*}
We now bound this difference uniformly.

Fix any $u\in U$. Then $u$ can be written as
\begin{align*}
u=\sum_{k=1}^K \alpha_k \mathbf u_{y_k}
\end{align*}
for some $K<\infty$, $\alpha_k\ge 0$, $\sum_{k=1}^K \alpha_k=1$, and $y_k\in\mathbb R$. For each $y_k$, let $x(y_k)\in\mathcal X$ denote a closest grid point. Define
\begin{align*}
\widetilde u \equiv \sum_{k=1}^K \alpha_k \mathbf u_{x(y_k)}\in \widetilde U.
\end{align*}

Then by the triangle inequality,
\begin{align*}
\|u-\widetilde u\|_2\ \le\ \sum_{k=1}^K \alpha_k\|\mathbf u_{y_k}-\mathbf u_{x(y_k)}\|_2\ \le\ \sup_{y\in\mathbb R}\|\mathbf u_y-\mathbf u_{x(y)}\|_2 ,
\end{align*}
where $x(y)$ is a closest grid point to $y$; in particular $x(y)=L$ for $y>L$ and $x(y)=-L$ for $y<-L$. By Parseval's formula and $w_j\le1$, $\|\mathbf u_y-\mathbf u_{x(y)}\|_2^2\le\int\varphi(t)\big(\varphi(t-y)-\varphi(t-x(y))\big)^2dt$. For $y\in[-L,L]$ this is at most $\mathcal{E}_2(L,\delta)^2$ by definition. For $y>L$ the right-hand-side equals $\int\varphi(t)\big(\varphi(t-y)-\varphi(t-L)\big)^2dt$, which we now show is at most $\mathcal{E}_1(L)^2$; the case $y<-L$ follows by symmetry.

We claim that for any $L>0$ and any $y>L$,
\begin{align*}
\int_{-\infty}^\infty \varphi(t)\big(\varphi(t-L)-\varphi(t-y)\big)^2dt
\le
\int_{-\infty}^\infty \varphi(t)\varphi(t-L)^2dt.
\end{align*}
To see this, define
\begin{align*}
A(a)\equiv \int_{-\infty}^\infty \varphi(t)\varphi(t-a)^2dt,
\qquad
B(a,b)\equiv \int_{-\infty}^\infty \varphi(t)\varphi(t-a)\varphi(t-b)\,dt.
\end{align*}
Then: $
\int_{-\infty}^\infty \varphi(t)\big(\varphi(t-L)-\varphi(t-y)\big)^2dt
=
A(L)+A(y)-2B(L,y)$. A direct calculation gives
\begin{align*}
A(a)=\frac{1}{2\pi\sqrt{3}}e^{-a^2/3},
\qquad
B(a,b)=\frac{1}{2\pi\sqrt{3}}e^{-(a^2-ab+b^2)/3}.
\end{align*}



Therefore
\begin{align*}
A(y)\leq 2B(L,y)
\iff
e^{-y^2/3}\leq 2e^{-(L^2-Ly+y^2)/3}
\iff
y\ge L-\frac{3\log 2}{L}.
\end{align*}
Since $L>0$ and $y>L$, this condition is automatically satisfied. Hence
\begin{align*}
A(L)+A(y)-2B(L,y)\leq A(L),
\end{align*}
which proves the claim. Combining this with the previous results:
$$0\leq \widetilde d_{w,J}(\theta)-d_{w,J}(\theta) \leq \max\{\mathcal{E}_1(L),\mathcal{E}_2(L,\delta)\}$$

Taking the sum (instead of the max) makes the inequality strict since both quantities are nonzero. So:
\begin{align*}
    \widetilde d_{w,J}(\theta) < d_{w,J}(\theta)+\mathcal{E}(L,\delta) \leq d_{w}(\theta)+\mathcal{E}(L,\delta)
\end{align*}

\subsection{Proof of Proposition \ref{prop:implementation}}\label{proof:prop:implementation}

{\bf Two Useful Facts}

First, since $\widetilde{d}_{w,J}$ is the distance to a convex set, it is $1$-Lipschitz and $\sqrt{m}|\widetilde{d}_{w,J}(\widehat{\theta}_n)-\widetilde{d}_{w,J}(\theta_0)| \leq ||\sqrt{m}W(\widehat{\theta}_n-\theta_0) ||_{\ell^2}$. By Lemma \ref{lem:clt_and_bootstrap}, $||\sqrt{m}W(\widehat{\theta}_n-\theta_0) ||_{\ell^2} = \mathcal{O}_p(1)$. Second, the numerical derivative is also $1$-Lipschitz in $Wh$ and $|\widetilde{\phi}'(h)|\leq\|Wh\|_{\ell^2}$. So, $\widetilde{cv}(\alpha)=\mathcal{O}_p(1)$ also by Lemma \ref{lem:clt_and_bootstrap}. 

    {\bf Pointwise Validity}

    Under the null hypothesis, $d_w(\theta_0)=0$ by Lemma \ref{lem:H0_dw}. So by Lemma \ref{lem:max_grid_approx_error}, $\tilde{d}_{w,J}(\theta_0) < \mathcal{E}(L,\delta)$. So:
    \begin{align*}
        \sqrt{m}\tilde{d}_{w,J}(\widehat{\theta}_n)-\widetilde{cv}(\alpha) -\sqrt{m}\mathcal{E}(L,\delta) \leq \sqrt{m}(\tilde{d}_{w,J}(\theta_0)-\mathcal{E}(L,\delta))+\mathcal{O}_p(1)\to_p -\infty
    \end{align*}
    \noindent and the test rejects with probability approaching zero. 

    {\bf\noindent Consistency against fixed alternatives}

    Since $\mathcal{D}_{\mathcal{X}}\subseteq\mathcal{D}$, then $\widetilde{d}_{w,J}(\theta)\ge d_{w,J}(\theta)$. The two facts proved above guarantee:
    \begin{align*}
        \sqrt{m} \widetilde{d}_{w,J}(\widehat{\theta}_n)-\widetilde{cv}(\alpha)-\sqrt{m} \mathcal{E}(L,\delta) \geq \sqrt{m}\left(d_{w,J}(\theta_0)-\mathcal{E}(L,\delta)\right)-\mathcal{O}_p(1) \to_p \infty
    \end{align*}

\subsection{Proof of Theorem \ref{thm:implementation_asymptotic}}\label{proof:thm:implementation_asymptotic}

As in the proof of Proposition \ref{prop:implementation}, $\sqrt{m}\,|\widetilde{d}_{w,J}(\widehat{\theta}_n)-\widetilde{d}_{w,J}(\theta_0)|\le\|\sqrt{m}\,W(\widehat{\theta}_n-\theta_0)\|_{\ell^2}=O_p(1)$ and $\widetilde{cv}(\alpha)=O_p(1)$. This does not depend on $J$ or the grid, so it holds along any sequences $J_m$, $L_m$, $\delta_m$.

{\bf\noindent Pointwise Validity} 

 Under the null hypothesis, $d_w(\theta_0)=0$ by  Lemma \ref{lem:H0_dw}. So Lemma \ref{lem:max_grid_approx_error} gives $\mathcal{E}_m-\widetilde{d}_{w,J_m}(\theta_0)\ge\min\{\mathcal{E}_1(L_m),\mathcal{E}_2(L_m,\delta_m)\}$. Hence
$$
  \sqrt{m} \widetilde{d}_{w,J_m}(\widehat{\theta}_n)-\widetilde{cv}(\alpha)-\sqrt{m} \mathcal{E}_m \leq -\sqrt{m} \min\big\{\mathcal{E}_1(L_m),\mathcal{E}_2(L_m,\delta_m)\big\}+\mathcal{O}_p(1)\ \to_p -\infty
$$

{\bf\noindent Consistency against fixed alternatives}

If $\mathbf{H}_0$ is false, then $d_w(\theta_0)>0$ by Lemma \ref{lem:H0_dw}.  We first bound the effect of truncation. Let $\omega_J\equiv(\sum_{j>J}w_j^2)^{1/2}$. By Lemma \ref{lem:bound_coeffs}, $|\theta_{0,j}|\le(2\pi)^{-1/2}$, and by Lemma \ref{lem:cj_Pi}, $|c_j(\Pi)|\le2^{-j/2}/\sqrt{4\pi}$, for every $j$ and $\Pi$. So for any $\Pi\in\mathcal{D}$, splitting the norm at $J$ and using $w_j\le1$,
\[
  \|W(\theta_0-\mathbf{c}(\Pi))\|_{\ell^2}\ \le\ \Big(\sum_{j=0}^Jw_j^2\big(\theta_{0,j}-c_j(\Pi)\big)^2\Big)^{1/2}+(2\pi)^{-1/2}\omega_J+2^{-J/2}/\sqrt{4\pi}.
\]
Taking the infimum over $\Pi$ gives $d_w(\theta_0)\le d_{w,J}(\theta_0)+(2\pi)^{-1/2}\omega_J+2^{-J/2}/\sqrt{4\pi}$. Since $\mathcal{D}_{\mathcal{X}}\subseteq\mathcal{D}$, $\widetilde{d}_{w,J}(\theta_0)\ge d_{w,J}(\theta_0)$. Combining these with the two facts,

$$
  \sqrt{m}\,\widetilde{d}_{w,J_m}(\widehat{\theta}_n)-\widetilde{cv}(\alpha)-\sqrt{m}\,\mathcal{E}_m\ \ge\ \sqrt{m}\,\Big(d_w(\theta_0)-(2\pi)^{-1/2}\omega_{J_m}-2^{-J_m/2}/\sqrt{4\pi}-\mathcal{E}_m\Big)-O_p(1)\ \to\ \infty
$$
 because $\omega_{J_m}\to0$ and $2^{-J_m/2}\to0$ as $J_m\to\infty$, and $\mathcal{E}_m\to0$.

\subsection{Proof of Lemma \ref{lem:bound_coeffs}}\label{proof:lem:bound_coeffs} 

 First we convert from the probabilist's Hermite polynomials $He_j(t)$ to  physicists'  version  $H_j(t)$. Recall that $He_j(t) = 2^{-j/2} H_j(t/\sqrt{2})$,
so $\psi_j(t) = \frac{1}{\sqrt{j!}}He_j(t)
= \frac{2^{-j/2}}{\sqrt{j!}}H_j(t/\sqrt{2})$. Hence:
  \begin{align*}
 \sup_{t\in\mathbb{R}}\left|\psi_j(t)\varphi(t)\right| &=\sup_{t\in\mathbb{R}}\left|\psi_j(t)e^{-t^2/2}/(\sqrt{2\pi})\right|\\
     &=\sup_{t\in\mathbb{R}} \left|H_j(t/\sqrt{2})2^{-j/2}e^{-t^2/2}/\sqrt{2\pi j!}\right|
\end{align*}

\noindent  Equation (2) of \cite{HermiteBound} says that the physicist's Hermite polynomials are bounded by:
    $|H_j(t)| \leq (2^j j!)^{1/2}e^{t^2/2}$. Therefore:
\begin{align*}
    \left|H_j(t/\sqrt{2})\right|2^{-j/2}e^{-t^2/2}/\sqrt{2\pi j!} &\leq (2^j j!)^{1/2}e^{t^2/4}2^{-j/2}e^{-t^2/2}/\sqrt{2\pi j!}=e^{-t^2/4}/\sqrt{2\pi}\leq \frac{1}{\sqrt{2\pi}}
\end{align*}

It is also immediate that $\sup_{t\in\mathbb{R}}\left|\chi_j(t)\varphi_{2}(t)\right| \leq \frac{1}{\sqrt{4\pi}} $. To see this, recall that  $\varphi_2(t)=\varphi(t/\sqrt{2})/\sqrt{2}$ and $\chi_j(t)=\psi_j(t/\sqrt{2})$ and therefore: $\chi_j(t)\varphi_2(t) =\psi_j(t/\sqrt{2})\varphi(t/\sqrt{2})/\sqrt{2}$. Now we apply the first claim:  $\sup_{t\in\mathbb{R}}\left|\chi_j(t)\varphi_{2}(t)\right| \leq\sup_{t\in\mathbb{R}}|\psi_j(t/\sqrt{2})\varphi(t/\sqrt{2})/\sqrt{2}|\leq \frac{1}{\sqrt{2}}\frac{1}{\sqrt{2\pi}}$.
 
\subsection{Lemma \ref{lem:absorb} and proof}\label{proof:lem:absorb}

\begin{lemma}\label{lem:absorb}
Let Assumption \ref{assum:weights} hold and assume that there is no selective reporting. Then the following are true:
\begin{enumerate}
\item[(i)]  If $T=H+W_0+Z$ with $Z\sim N(0,1)$ and $Z \independent(H,W_0)$, then $d_w(\mathbb{E}[\xi(T)])=0$. 

\item[(ii)] If $T=L(H+Z)$ with $Z\sim N(0,1)$, $\mathbb{P}\left[L\ge 1\right]=1$, and $Z \independent (L,H)$, then $d_w(\mathbb{E}[\xi(T)])=0$.

\item[(iii)] If studies differ in characteristics $S$, possibly including $H$, then $d_w(\mathbb{E}[\xi(T)])\le E[d_w(\mathbb{E}[\xi(T)\mid S])]$.
\end{enumerate}
\end{lemma}

\begin{proof}

(i)  $T=H+W_0+Z$ satisfies $\mathbf{H}_0$ with $H'=H+W_0$. So $\mathbb{E}[\xi(T)]=c(\Pi')$ for the distribution $\Pi'$ of $H+W_0$, and $d_w(\mathbb{E}[\xi(T)])\leq d(\mathbb{E}[\xi(T)])=0$. 

(ii) For $L \geq1$, let $Z'\sim N(0,1)$ be independent of $(H,Z,L)$. Let $W_0=(L-1)H+\sqrt{L^2-1}Z'$. Since $Z\independent (H,W_0)$,  by (i) the distribution of  $H+W_0+Z$ satisfies $\mathbf{{H}}_0$. So we need only show that $T\stackrel{d}{=}H+W_0+Z$. Conditionally, $L Z\mid (L,H)\sim N(0,L^2) $ and $Z+\sqrt{L^2-1}\,Z'\mid (L,H)\sim N(0,L^2)$ because $(L,H)\independent (Z,Z')$. So $(LH+LZ \stackrel{d}{=} H+W_0+Z) \mid L,H$. Since this holds conditionally, it must also hold unconditionally. So $T \stackrel{d}{=}H'+Z$ with $H' = LH+\sqrt{L^2-1}Z'$.

(iii) $d_w$ is the distance from $W\theta$ to a convex set, and hence is a convex function. Since $\mathbb{E}[\xi_j(T)]=\mathbb{E}\big[\mathbb{E}[\xi_j(T)\mid S]\big]$ for every $j$, Jensen's inequality gives the result.
\end{proof}

\subsection{Proof of Proposition \ref{prop:rounding}}\label{proof:prop:rounding}

By Lemma \ref{lem:absorb}(iii), it suffices to construct a bound for a homogeneous literature with fixed true effect $h$ and fixed $\kappa$. Let $T^*=h+Z$, let $r(T^*)$ be $T^*$ rounded to the nearest multiple of $\kappa$, and let $\tilde U$ be uniform on $[-\kappa/2,\kappa/2]$ and independent of $Z$. By Lemma \ref{lem:absorb}(i), $T^*+\tilde U$ satisfies ${\mathbf{H}_0}$, so it suffices to compare $r(T^*)$ with it. Let $\tilde f$ be the density of $T^*+\tilde U$ and let $f^*$ be the density of $T^*$. Then for integer $k$, we have $\tilde f (k \kappa) = \frac{1}{\kappa}\int_{-\kappa/2}^{\kappa/2}f^*(k\kappa-x)dx = \frac{1}{\kappa}\mathbb{P}(|T^*-k\kappa|\leq \kappa/2)$. Notice also that $\mathbb{P}(r(T^*)=k\kappa)=\mathbb{P}(|T^*-k\kappa| \leq \kappa/2)$. So combining: $\mathbb{P}(r(T^*)=k\kappa) = \kappa \tilde f (k \kappa)$.

Let $\theta_{r,j} \equiv \mathbb{E}[\xi_j(r(T^*))]$ and let $\theta_{Y,j}\equiv \mathbb{E}[\xi_j(T^* +\tilde U)]$. Let $u_\beta(t)=\sum_j\beta_jw_j\xi_j(t)$ and $F(t)=u_\beta(t)\tilde f(t)$. 

Since $He_n\varphi=(-1)^n\varphi^{(n)}$, we have
$(\varphi\psi_j)^{(i)}=(-1)^i\sqrt{(j+i)!/j!}\,\varphi\psi_{j+i}$. By Lemma \ref{lem:bound_coeffs} and
$\int(\varphi\psi_n)^2\le\sup_t\varphi(t)\int\psi_n^2\varphi=(2\pi)^{-1/2}$, the
sup and $L^2$ norms of $\xi_j^{(i)}$ are at most $\sqrt{(j+i)!/j!}\le 6^3(1+j)^3$ for
$i\le 6$; this also holds after shifting and symmetrizing, since shifts and reflections
preserve both norms. By Cauchy--Schwarz and $\sum_j w_j^2(1+j)^6<\infty$, the series
defining $u_\beta$ and its termwise derivatives up to order six converge uniformly and in
$L^2$. Hence $u_\beta\in C^6$ with bounded derivatives
$u_\beta^{(i)}=\sum_j\beta_jw_j\xi_j^{(i)}$, and $\Big\|\sum_{j=0}^\infty \beta_jw_j\xi_j^{(i)}\Big\|_2   < \infty$. 

Now we upper bound the difference between the rounded $\theta_r$ and the set of honest distributions:
\begin{align*}
    d_w(\theta_r)&\leq \sqrt{\sum_{j=0}^\infty w_j^2(\theta_{r,j}-\theta_{Y,j})^2} = \sup_{\|\beta\|_{\ell^2}=1}\sum_{j=0}^\infty\beta_jw_j(\theta_{r,j}-\theta_{Y,j})\\
    &= \sup_{\|\beta\|_{\ell^2}=1}\sum_{j=0}^\infty\beta_jw_j(\mathbb{E}[\xi_j(r(T^*))] - \mathbb{E}[\xi_j(T^*+\tilde U)])\\
    &= \sup_{\|\beta\|_{\ell^2}=1}\left(\mathbb{E}\left[\sum_{j=0}^\infty\beta_jw_j\xi_j(r(T^*))\right] - \mathbb{E}\left[\sum_{j=0}^\infty\beta_jw_j\xi_j(T^*+\tilde U)\right]\right)\\
    &= \sup_{\|\beta\|_{\ell^2}=1}\left(E[u_\beta(r(T^*))]-E[u_\beta(T^*+\tilde U)]\right)
\end{align*}

Evaluate the expectations, then take the Fourier transform: $\hat F(\omega) \equiv\int_{-\infty}^\infty e^{-i\omega t}F(t)dt $. Since each $\tilde f^{(m)}$ has
Gaussian tails, every $F^{(k)}=\sum_{i\le k}\binom{k}{i}u_\beta^{(i)}\tilde f^{(k-i)}$ with
$k\le 6$ is integrable and vanishes at $\pm\infty$. This justifies the six integrations by
parts below, and with the resulting bound $|\hat F(\omega)|\le|\omega|^{-6}\|F^{(6)}\|_1$,
both $F$ and $\hat F$ decay fast enough for the Poisson summation formula. Now apply the Poisson summation formula:
\begin{align*}
E[u_\beta(r(T^*))]-E[u_\beta(T^*+\tilde U)] &= \kappa\sum_kF(k\kappa)-\int_{-\infty}^\infty F(t)dt\\
   &=\sum_{\ell}\hat F(2\pi\ell/\kappa)-\hat F(0) =\sum_{\ell\ne0}\hat F(2\pi\ell/\kappa)
\end{align*}

Integrating by parts six times gives $|\hat F(\omega)|\leq|\omega|^{-6}\int_{-\infty}^\infty|F^{(6)}(t)|dt$, and $\sum_{\ell\ne0}(\kappa/2\pi\ell)^6=\kappa^6/30240$. Substituting:
\begin{align*}
    \sum_{\ell\ne0} |\hat F(2\pi\ell/\kappa)| &\leq \sum_{\ell\ne0}(\kappa/2\ell \pi)^{6}\int_{-\infty}^\infty|F^{(6)}(t)|dt = \frac{\kappa^6}{30240}\int_{-\infty}^\infty|F^{(6)}(t)|dt
\end{align*}

By the Leibniz rule:$ F^{(6)}(t) = \sum_{i=0}^6 \binom6i u_{\beta}^{(i)}(t)\tilde{f}^{(6-i)}(t)$. 
Taking absolute values, integrating and taking the triangle inequality and then Cauchy-Schwarz:
\begin{align*}
   \int_{-\infty}^\infty|F^{(6)}(t)|dt \leq  \sum_{i=0}^6 \binom6i \int_{-\infty}^\infty \left|u_{\beta}^{(i)}(t)\right|\,\left|\tilde{f}^{(6-i)}(t)\right|dt \leq \sum_{i=0}^6 \binom6i \left|\left|u_{\beta}^{(i)}\right|\right|_2\,\left|\left|\tilde{f}^{(6-i)}\right|\right|_2 
\end{align*}

To bound $\left|\left|\tilde{f}^{(6-i)}\right|\right|_2 $, recall that we fixed $h$:
\begin{align*}
    \left|\left|\tilde{f}^{(6-i)}(y)\right|\right|_2 = \left|\left|\frac{1}{\kappa}\int_{-\kappa/2}^{\kappa/2} \varphi^{(6-i)}(y-h-s)ds\right|\right|_2 \leq \left|\left|\varphi^{(6-i)}(y)\right|\right|_2
\end{align*}

Putting these together:
\begin{align}
    d_w(\theta_r) \leq  \frac{\kappa^6}{30240} \sum_{i=0}^6 \binom6i \left|\left|\varphi^{(6-i)}(y)\right|\right|_2 \sup_{\|\beta\|_{\ell^2}=1}\Big\|\sum_{j=0}^\infty \beta_jw_j\xi_j^{(i)}\Big\|_2  
\end{align}
The right hand side is $\kappa^6$ times a known deterministic quantity that can be computed numerically.  The study is therefore within $C\kappa^6$ of the null set, with: $$C= \frac{1}{30240} \sum_{i=0}^6 \binom6i \left|\left|\varphi^{(6-i)}(y)\right|\right|_2 \sup_{\|\beta\|_{\ell^2}=1}\Big\|\sum_{j=0}^\infty \beta_jw_j\xi_j^{(i)}\Big\|_2$$ If independent noise $V$ is added, instead compare $r(T^*)+V$ with $T^*+\tilde U+V$: for each value $v$ of $V$, the argument applies to $u_\beta(\cdot+v)$, which has the same values of $\sup_{\|\beta\|_{\ell^2}=1}\Big\|\sum_{j=0}^\infty \beta_jw_j\xi_j^{(i)}\Big\|_2$.

\subsection{Proof of Proposition \ref{prop:ratio}}\label{proof:prop:ratio}

(i) Rounding the standard error always gives $|s-\hat\sigma|\le b/2$, and $s>0$ forces $s\ge b$, so $\hat\sigma\ge s_{\min}\ge b/2>0$. Let $L=\hat\sigma/s_{\min}\ge1$ which by the hypothesis must be independent of $Z|H$. So $\tilde{T}=L(H+Z)$ where $L\geq 1$ almost surely. Carrying out the construction in the proof of Lemma \ref{lem:absorb}(ii) conditional on $(H,a,b,\hat\sigma)$ instead of $(H,L)$ gives $(\tilde T,\kappa)\stackrel{d}{=}(H'+Z,\kappa)$ with $Z\independent(H',\kappa)$. So $\tilde{T}$ satisfies $\mathbf{H}_0$.  (ii) Since $m=a\,\mathrm{round}(\hat\mu/a)$, $m/s_{\min}=\kappa\,\mathrm{round}(\tilde{T}/\kappa)$, so $T_{der}$ has the same distribution as $H'+Z$ rounded to the nearest multiple of $\kappa$ plus the independent noise $\kappa U$. Since $\kappa\independent Z\mid H'$, Proposition \ref{prop:rounding} gives $d_w(\mathbb{E}[\xi(T_{der})])\le C\mathbb{E}[\kappa^6]$.

\clearpage
\begin{center}
    {\huge Online Appendix}
\end{center}

\section{Extended Maximization $p$-hacking Example}\label{sec:extended_max}

This appendix extends the example of undetectable $p$-hacking from Section \ref{sec:max_hacking}. The purpose of the extension is to show that the undetectability of maximization $p$-hacking is not merely a knife-edge case.

\begin{proposition}
    Suppose that the researcher observes $k$ $t$-scores $\{T_i\}_{i=1}^k$. Each $T_i=H+Z_i$ where $Z_i = \sqrt{\rho}Z_0+\sqrt{1-\rho}W_i$. So conditional on $H$ the $T_i$ are jointly normal centered on $H$ with variance $1$ and common correlation $\rho \in [0,1]$. The distribution of true effects has the following form: $H = A+S$ where $A \independent S$, $A$ can have any distribution, and $S \sim N(0,s^2)$ for $s > 0$. The researcher $p$-hacks by reporting only $T^* = \max_{i\in 1,\cdots k}T_i$. 
    
    Then  $p$-hacking is undetectable (i.e. $\mathbf{H_0}$ is satisfied for the distribution of $T^*$) if: $$s^2 \geq (1-\rho)(1-1/k)$$

\end{proposition}
\begin{proof}

Define $D \equiv \max_i W_i-\overline{W}$, which must be non-negative. Adding and subtracting $\sqrt{1-\rho}\overline{W}$:
$$ T^* = A+S+ \sqrt{\rho}Z_0+\sqrt{1-\rho}D+\sqrt{1-\rho} \overline{W} $$

Since the $W_i$ are iid normal they are jointly normal and $\overline{W}\independent D$. Moreover, $(Z_0,W) \independent (A,S)$ by construction. Thus, $A,S+ \sqrt{\rho}Z_0+\sqrt{1-\rho}\overline{W},D$ are mutually  independent and $S+ \sqrt{\rho}Z_0+\sqrt{1-\rho}\overline{W} \sim N(0,s^2+\gamma)$ where $\gamma = \rho+(1-\rho)/k$. 

If $s^2+\gamma \geq 1$, then $T^*$ is the sum of a normal component $S+ \sqrt{\rho}Z_0+\sqrt{1-\rho}\overline{W} $ and another component $A+\sqrt{1-\rho} D $ that is independent of it. So ${\mathbf{H}_0}$ holds. Notice that $s^2+\gamma \geq 1 \iff s^2 \geq (1-\rho)(1-1/k)$.

\end{proof}

\section{Local Untestability}\label{app:local_untestability}

In this technical appendix we show that no locally valid test of ${\mathbf{H}_0}$ has local power greater than size in any direction when the support of $\Pi_0$ contains an interval. This has several important implications. For theory, the lesson is that we cannot rank tests by local power. Thus, our goal becomes finding a test that is consistent against all fixed detectable alternatives. The projection test satisfies this criterion. For practice, local untestability means heuristically that small amounts of selective reporting are practically undetectable to all tests. 

\subsection{Setup}
For simplicity we will restrict ourselves to settings with one $t$-score per article in this section.  Following Definition 2.1 of \cite{chensantos},  a \emph{path} through $f_{T}$ is a family of densities $f_{s}$, $s\in[0,1)$ such that $f_0=f_T$ at $s=0$. It has \emph{score} $g$ if:
\begin{equation}\label{eq:pathdef}
    \lim_{s\downarrow 0}\int_{-\infty}^\infty\left[ \frac{1}{s}\left(f_{s}(t)^{1/2}- f_{0}(t)^{1/2}\right)-\frac{1}{2}g(t)f_{0}(t)^{1/2} \right]^2dt=0
\end{equation}

Define the set $L_0^2(f_{0})$ as:
\begin{align}
    L_0^2(f_{0}) \equiv \left\{g \: : \: \mathbb{R}\to \mathbb{R},\: \int_{-\infty}^\infty g(t)f_{0}(t)dt=0,\: \int_{-\infty}^\infty g(t)^2f_{0}(t)dt < \infty\right\}
\end{align}

And define the tangent set at $f_0$ as:
\begin{align}
     \mathcal{T}(f_{0}) \equiv  \left\{g\in  L_0^2(f_{0}),\: (\ref{eq:pathdef})\text{  holds for $g$ and some path with $ f_{s}$ all satisfying $\mathbf{H}_0$} \right\}
\end{align}
\noindent We say that $f_{0}$ is {\it locally just identified} if every $g \in L_0^2(f_{0})$ is a limit of the elements of the tangent set, i.e. $\mathcal{T}(f_{0})$ is dense in $L_0^2(f_{0})$ (Definition 2.2 of \cite{chensantos}). Intuitively, this means that every tiny step away from $f_{0}$ can be first-order mimicked from inside the null set.

We are interested in a specific set of paths inside ${\mathbf{H}_0}$ with scores that are known to be dense in $L_0^2(f_{0})$. To construct such a set we will use \citep[Example 25.35]{vanderVaart1998}. Let $\B_0$ be the set of bounded measurable functions $b:\R\to\R$ with $\int_{-\infty}^\infty b(h)\,d\Pi_0(h)=0$. For $b\in\B_0$ and $|s|<1/\norm{b}_\infty$, let $\Pi_{s,b}$ be the probability distribution defined by $d\Pi_{s,b}(h)=(1+s b(h))\,d\Pi_0(h)$. So its $t$-curve is $f_{s,b}(t)= \int\varphi(t-h)d\Pi_{s,b}(h)=f_{0}(t)\,(1+s \,(Ab)(t))$, where
\[
(Ab)(t)=\mathbb{E}_{\Pi_0}[b(H)\mid T=t]=\frac{\int b(h)\,\varphi(t-h)\,d\Pi_0(h)}{f_0(t)},\qquad |Ab|\le\norm{b}_\infty .
\]
Because $Ab$ is bounded, $\sqrt{1+s \,Ab}=1+\frac{s}{2}Ab+O(s^2)$ uniformly, so $s\mapsto f_{s,b}$ satisfies (\ref{eq:pathdef}) at $s=0$ with score $Ab$. When $s \geq 1/(2||b||_\infty)$, set the path to its value there, leaving the score unchanged and making $f_{s,b}$ a path in the sense of Definition 2.1 of \cite{chensantos}. This is the score operator of the mixture model in \citep[Example 25.35]{vanderVaart1998}. Notice that $f_{s,b}$ satisfies $\mathbf{H}_0$  with ``alibi" $\Pi_{s,b}$ for all $s$. So every Ab belongs to $\mathcal{T}(f_0)$.

\subsection{Result}

A key fact about the normal mixture model specified by $\mathbf{H}_0$ is that it is locally just identified in the sense of \cite{chensantos} when the support of $\Pi_0$ contains an interval. This observation is a special case of \citep[Example 25.35]{vanderVaart1998}  with $\varphi(t-h) = e^{-h^2/2}\varphi(t)e^{ht}$ and noting that $b\mapsto Ab$ is a contraction and bounded functions are dense. 

\begin{fact}[Local just identification]\label{lem:LJI}
Suppose $\supp(\Pi_0)$ contains a nondegenerate interval. Then the scores $\{Ab:b\in\B_0\}$ are dense in $L^2_0(f_0)$: for every $g\in L^2_0(f_0)$ and $\epsilon>0$ there is $b\in\B_0$ with $\left[\int \,\{g(t)-(Ab)(t)\}^2{f_0(t)}dt\right]^{1/2}<\epsilon$. Moreover, $f_0$ is locally just identified.  
\end{fact}

This mathematical fact has a very important implication for detecting selective reporting: small deviations from $\mathbf{H}_0$ are extremely hard to detect.  Proposition \ref{prop:notest} shows that when the distribution of true effects has support that contains an interval, then no locally valid test has any local power.

\begin{proposition}\label{prop:notest}
Suppose $\supp(\Pi_0)$ contains a nondegenerate interval. Consider a test and write
$\mathbb{P}_{f}(\text{reject})$ for its rejection probability when $T_1,\dots,T_n$ are i.i.d.\
with density $f$. Suppose the test controls size along mixing paths: $\limsup_{n\to \infty}\mathbb{P}_{f_{n^{-1/2},b}}(\text{reject})\le\alpha$ for every $b\in\B_0$. Then $\limsup_{n\to \infty} \mathbb{P}_{f_{n^{-1/2}}}(\text{reject})\le\alpha$ for every path $f_{s}$ through $f_0$ that has a score, whether or not it consists of densities that satisfy ${\mathbf{H}_0}$. Proof: Section \ref{proof:prop:notest}.
\end{proposition}

To see how this applies to $p$-hacking, notice that every selective reporting process that picks from a set of finitely many honest $t$-scores or filters them can be written as such a path satisfying (\ref{eq:pathdef}). Suppose that some fraction $1-\tau \in (0,1]$ of  researchers are honest and report draws from $f_T$. But some fraction $\tau \in [0,1)$ of researchers  instead report draws from density $\tilde{f}(t)$. Assume that $\tilde{f}(t)\le Cf_0(t)$ for all $t$ and some $C<\infty$. Then $f_{\tau}(t)=(1-\tau)f_T(t)+\tau \tilde{f}(t)=f_T(t)[1+\tau \{\tilde{f}(t)/f_T(t)-1\}]$ is a path with score $g(t)=\tilde{f}(t)/f_T(t)-1$, because the perturbation is bounded and $f_T=f_0$. Two leading cases of selective reporting satisfy this:
\begin{itemize}
\item Suppose that a study reports one of $K$ honest draws $T^*\in \{T_1,\dots,T_K\}$, each with density $f_T$. No matter what the selection rule is, the reported density  satisfies $\tilde{f}(t)\leq Kf_T(t)$, because $P(T^*\in B)\le\sum_{j=1}^KP(T_j\in B)$ for every set $B$.
\item Publication bias with $P(R)>0$: $\tilde{f}(t)=f_0(t)P(R\mid T=t)/P(R)\le f_T(t)/P(R)$.
\end{itemize}
So $f_{\tau}$ describes a literature in which a share $\tau$ of studies is selectively reported. By Proposition \ref{prop:notest}, no test that controls size along mixing paths can detect such selection when its share is of order $n^{-1/2}$.

\begin{remark}
    One proof route for Proposition \ref{prop:notest} is to verify the conditions of \cite{chensantos} Theorem 3.1(ii) including their Assumption 2.1(ii). It turns out that proving the desired result directly is faster in our setting, so we have done so. On a mathematical level, we are paralleling the known machinery of \cite{chensantos} and applying \citep[Example 25.35]{vanderVaart1998}. Our contribution is to apply this framework to the $p$-hacking detection problem and to point out its important implications.
\end{remark}

\subsection{Proof of Proposition \ref{prop:notest}}\label{proof:prop:notest}

\begin{proof}
For densities $p$ and $q$, on one coordinate or on several, write
\[
\Hel(p,q)=\Big\{\int\big(\sqrt{p(t)}-\sqrt{q(t)}\big)^2dt\Big\}^{1/2},\qquad V(p,q)=\frac12\int|p(t)-q(t)|\,dt
\]
for the Hellinger and total variation distances. We use two standard facts about them (\cite{Tsybakov2009}, Section 2.4). First, $V(p,q)\le\Hel(p,q)$ (\cite{Tsybakov2009}, Lemma 2.3). Hence, for any measurable function $u(t)$ with $0\le u(t)\le1$, such as a test applied to the same data under $p$ and under $q$,
\begin{equation}\label{eq:hellinger_1}
    \Big|\int u(t)\{p(t)-q(t)\}\,dt\Big|\le\int|p(t)-q(t)|\,dt=2V(p,q)\le2\,\Hel(p,q)
\end{equation}

Second, for the joint densities $p^{(n)}(x_1,\dots,x_n)=\prod_ip(x_i)$ and $q^{(n)}(x_1,\dots,x_n)=\prod_iq(x_i)$ of $n$ i.i.d.\ draws,
\begin{equation}\label{eq:Hel_pn}
    \Hel(p^{(n)},q^{(n)})^2=2\big[1-\{1-\Hel(p,q)^2/2\}^n\big]\le n\,\Hel(p,q)^2.
\end{equation}
\noindent The equality is the product formula for the Hellinger distance (\cite{Tsybakov2009}, Section 2.4), and the inequality is Bernoulli's inequality, $(1-a)^n\ge1-na$ for $0\le a\le1$, applied with $a=\Hel(p,q)^2/2$.

Let $f_{s}$ be a path through $f_0$ with score $g$, and let $\epsilon>0$. By Fact~\ref{lem:LJI} there is $b\in\B_0$ with $\left[\int \,\{g(t)-(Ab)(t)\}^2{f_0(t)}dt\right]^{1/2}\le\epsilon$. Define the remainders
\[
\begin{aligned}
r_{1,s}(t)&=\sqrt{f_{s}(t)}-\sqrt{f_0(t)}-\frac{s}{2}\,g(t)\sqrt{f_0(t)}\\
r_{2,s}(t)&=\sqrt{f_{s,b}(t)}-\sqrt{f_0(t)}-\frac{s}{2}\,(Ab)(t)\sqrt{f_0(t)}
\end{aligned}
\]
The definition of a score (\ref{eq:pathdef}) says exactly that $\int r_{1,s}(t)^2\,dt=o(s^2)$, because $f_{s}$ has score $g$, and that $\int r_{2,s}(t)^2\,dt=o(s^2)$, because the mixing path $f_{s,b}$ has score $Ab$. Subtracting the second line from the first,
\[
\sqrt{f_{s}(t)}-\sqrt{f_{s,b}(t)}=\frac{s}{2}\,\{g(t)-(Ab)(t)\}\sqrt{f_0(t)}+r_{1,s}(t)-r_{2,s}(t)
\]
\noindent Square-integrating both sides and applying the triangle inequality:
$$\left[\int \left( \sqrt{f_{s}(t)}-\sqrt{f_{s,b}(t)}\right)^2dt\right]^{1/2}\leq  \frac{s}{2}\left[\int \,\{g(t)-(Ab)(t)\}^2{f_0(t)}dt\right]^{1/2} +o(s)$$

Substituting in definitions:
\begin{equation}\label{eq:Handnorm}
    \Hel(f_{s},f_{s,b})\leq  \frac{s}{2}\left[\int \,\{g(t)-(Ab)(t)\}^2{f_0(t)}dt\right]^{1/2} +o(s)
\end{equation}

Write $f_s^{(n)}$ and $f_{s,b}^{(n)}$ for the joint densities of n i.i.d. draws from $f_s$ and $f_{s,b}.$ Combining (\ref{eq:Hel_pn}) and (\ref{eq:Handnorm}):
\begin{equation}
 \Hel(f_{s}^{(n)},f_{s,b}^{(n)})  \leq  \frac{s \sqrt{n}}{2}\left[\int \,\{g(t)-(Ab)(t)\}^2{f_0(t)}dt\right]^{1/2}  +o(s \sqrt{n})
\end{equation}
 Now set $s=n^{-1/2}$:
 \begin{equation}
 \Hel(f_{s}^{(n)},f_{s,b}^{(n)})  \leq  \frac{1}{2}\left[\int \,\{g(t)-(Ab)(t)\}^2{f_0(t)}dt\right]^{1/2}  +o(1)
\end{equation}

Recall that Fact \ref{lem:LJI} guaranteed that $\left[\int \,\{g(t)-(Ab)(t)\}^2{f_0(t)}dt\right]^{1/2} \leq \epsilon$ and:
 \begin{equation}\label{eq:controlhel}
   \Hel(f_{s}^{(n)},f_{s,b}^{(n)})  \leq  \epsilon +o(1)
\end{equation}

Let $u(t_1,\dots,t_n)\in[0,1]$ be the probability that the test rejects when the
observed $t$-statistics are $t_1,\dots,t_n$, averaging over any bootstrap draws. Then
\[
P_f(\text{reject})=\int\!\cdots\!\int u(t_1,\dots,t_n)\,f(t_1)\cdots f(t_n)\,dt_1\cdots dt_n = \int u({t})f^{(n)}(t)dt
\]

By (\ref{eq:hellinger_1}) and (\ref{eq:controlhel}) with $s=n^{-1/2}$:
\begin{align*}
   \left|\mathbb{P}_{f_{n^{-1/2}}}(\text{reject})-\mathbb{P}_{f_{n^{-1/2},b}}(\text{reject})\right| &=\Big|\int u(t)f_{s}^{(n)}(t)dt-\int u(t)f_{s,b}^{(n)}(t)dt\Big|\\
    &\le2\,\Hel(f_{s}^{(n)},f_{s,b}^{(n)})\\
    &\leq  2\epsilon +o(1)
\end{align*}

Taking the limsup:
\begin{align*}
  \limsup_{n\to \infty} \mathbb{P}_{f_{n^{-1/2}}}(\text{reject}) &\leq \limsup_{n\to \infty} \mathbb{P}_{f_{n^{-1/2},b}}(\text{reject})+ 2\epsilon 
\end{align*}

 Since we assumed size control along mixing paths: $ \limsup_{n\to \infty}\mathbb{P}_{f_{n^{-1/2},b}}(\text{reject}) \leq \alpha$.  Since $\epsilon$ is arbitrary, the claim follows:
\begin{align*}
  \limsup_{n\to \infty} \mathbb{P}_{f_{n^{-1/2}}}(\text{reject}) &\leq \alpha
\end{align*}
\end{proof}

\clearpage
\section{Additional Empirical and Simulation Results}

\begin{table}[htbp]
\centering
\caption{The projection test with $J_0=30$ and without the shift, under four treatments of rounding, in economics (keeping coefficients reported as zero) and medicine}
\label{tab:rounding_appendix_keepmu0_drop_se_mi_cochrane}
\small
\resizebox{\ifdim\width>\linewidth\linewidth\else\width\fi}{!}{%
\begin{tabular}{lccc}
\toprule
 & \multicolumn{2}{c}{Projection, $J_0=30$} & \multicolumn{1}{c}{Projection, no shift} \\
\cmidrule(lr){2-3} \cmidrule(lr){4-4}
Treatment & $p$ & $10^3\widehat{B}$ & $p$ \\
\midrule
\multicolumn{4}{l}{\textit{RCT}: 6,795 $t$-statistics in 141 articles; rounding allowance $<0.1$} \\
None & \textbf{.00} &  & \textbf{.03} \\
Unif. & \textbf{.00} &  & \textbf{.03} \\
$s_{\min}$ & \textbf{.00} & 2.4 & \textbf{.03} \\
$s_{\min}$+noise & \textbf{.00} & 2.4 & \textbf{.03} \\
\addlinespace
\multicolumn{4}{l}{\textit{IV}: 4,619 $t$-statistics in 263 articles; rounding allowance $<0.1$} \\
None & \textbf{.00} &  & \textbf{.03} \\
Unif. & \textbf{.00} &  & \textbf{.03} \\
$s_{\min}$ & \textbf{.00} & 3.5 & \textbf{.03} \\
$s_{\min}$+noise & \textbf{.00} & 3.6 & \textbf{.03} \\
\addlinespace
\multicolumn{4}{l}{\textit{DID}: 4,966 $t$-statistics in 218 articles; rounding allowance $<0.1$} \\
None & \textbf{.01} &  & .09 \\
Unif. & \textbf{.02} &  & .10 \\
$s_{\min}$ & \textbf{.02} & 1.0 & .13 \\
$s_{\min}$+noise & \textbf{.02} & 1.0 & .12 \\
\addlinespace
\multicolumn{4}{l}{\textit{RDD}: 2,775 $t$-statistics in 80 articles; rounding allowance $<0.1$} \\
None & \textbf{.03} &  & .26 \\
Unif. & \textbf{.03} &  & .26 \\
$s_{\min}$ & \textbf{.01} & 1.9 & .20 \\
$s_{\min}$+noise & \textbf{.02} & 1.8 & .20 \\
\midrule
\multicolumn{4}{l}{\textit{Cochrane, all trials}: 17,680 outcomes in 5,310 trials; rounding allowance 0.2} \\
None & .61 &  & .54 \\
Unif. & .64 &  & .54 \\
$s_{\min}$ & .61 & 0 & .48 \\
$s_{\min}$+noise & .64 & 0 & .49 \\
\bottomrule
\end{tabular}}

\medskip
\begin{minipage}{\linewidth}
\footnotesize
\textit{Notes:} $p$-values for the samples and treatments of Table~\ref{tab:rounding_ekw_keepmu0_drop_se_mi_cochrane}.
$J_0=30$: weights $w_j=\min\{1,2^{30-j}\}$; $10^3\widehat{B}$ and the rounding allowance (first line of each block, $0.31\,\overline{\kappa^6}$) are for $J_0=30$.
No shift: $t$-statistics not shifted by 1.96 (Remark \ref{rem:shifting}); no allowance is computed for it, so its rejections are not certified against rounding.
Both cluster by article or trial and use 1,000 bootstrap repetitions.
Bold: $p<.05$; in the $s_{\min}$ rows, the $J_0=30$ test only when $\widehat{B}$ exceeds the allowance.
Run time per design of \cite{bb}, on one core: 8.1 s for the projection test,  2.8 s for the EKW tests.
\end{minipage}
\end{table}

\begin{table}[htbp]
\centering
\caption{Size and power of every test against outcome switching in Cochrane trials}
\label{tab:cochrane_power_all_tests}
\footnotesize
\setlength{\tabcolsep}{3pt}
\resizebox{\ifdim\width>\linewidth\linewidth\else\width\fi}{!}{%
\begin{tabular}{lcccccccccccc}
\toprule
 & \multicolumn{4}{c}{Size ($\tau = 0$)} & \multicolumn{4}{c}{Most significant ($\tau = 1$)} & \multicolumn{4}{c}{Threshold ($\tau = 1$)} \\
\cmidrule(lr){2-5} \cmidrule(lr){6-9} \cmidrule(lr){10-13}
Test & None & Unif. & $s_{\min}$ & $+$noise & None & Unif. & $s_{\min}$ & $+$noise & None & Unif. & $s_{\min}$ & $+$noise \\
\midrule
Projection test & .028 & .030 & .030 & .022 & .782 & .754 & .744 & .710 & .992 & .994 & .988 & .990 \\
Projection test, $p \le .05$ only & .028 & .030 & .034 & .028 & .782 & .754 & .756 & .730 & .992 & .994 & .988 & .990 \\
\addlinespace
EKW: LCM & .000 & .000 & .000 & .000 & .000 & .000 & .000 & .000 & .000 & .000 & .000 & .000 \\
EKW: Fisher & .000 & .000 & .000 & .000 & .000 & .000 & .000 & .000 & .000 & .000 & .000 & .000 \\
EKW: density discontinuity & .040 & .052 & \textbf{.070} & .044 & .024 & .050 & .018 & .048 & .012 & .028 & .030 & .040 \\
EKW: binomial & .032 & .014 & .022 & .014 & .010 & .008 & .012 & .008 & .012 & .010 & .004 & .008 \\
EKW: CS1 & .036 & .034 & .052 & .022 & .006 & .016 & .004 & .016 & .028 & .016 & .016 & .022 \\
EKW: CS2B & .060 & .052 & \textbf{.074} & .042 & .020 & .020 & .012 & .028 & .142 & .060 & .090 & .060 \\
EKW: smallest $p$-value & \textbf{.128} & \textbf{.122} & \textbf{.168} & \textbf{.096} & .054 & .086 & .042 & .080 & .162 & .094 & .122 & .112 \\
\addlinespace
$p$-curve (Stouffer) & .000 & .000 & .000 & .000 & .000 & .000 & .000 & .000 & .000 & .000 & .000 & .000 \\
EKW or $p$-curve: smallest $p$-value & \textbf{.128} & \textbf{.122} & \textbf{.168} & \textbf{.096} & .054 & .086 & .042 & .080 & .162 & .094 & .122 & .112 \\
\bottomrule
\end{tabular}}
\par\smallskip
\begin{minipage}{\linewidth}\footnotesize\textit{Notes:} Rejection rates of 5\%-level tests over 500 repetitions, with no switching ($\tau = 0$: size) and when every trial switches ($\tau = 1$). Each meta-sample draws 1,500 of the 3,282 randomized trials in the Cochrane Database of Systematic Reviews (Schwab's data set, as in van Zwet, Schwab and Senn 2021) with at least 2 continuous efficacy outcomes of at least 100 participants, without replacement, and reports one outcome per trial. A trial chooses among its first 8 outcomes in review order. With probability $1 - \tau$ it reports one chosen at random (its pre-specified outcome); with probability $\tau$ it reports its most significant outcome (Most significant) or, when the pre-specified outcome has $|z| \le 1.96$, the more significant of it and one other outcome (Threshold). Trials select on unrounded $z$: each repetition draws the means and standard deviations uniformly within their rounding cells. Columns within each group: the analyst's treatment of rounding. None: the printed $z$. Unif.: means and standard deviations drawn uniformly within their rounding cells. $s_{\min}$: the mean difference over $SE_{\min}$, the standard error from the smallest standard deviations consistent with the printed ones. $+$noise: the same with uniform noise at each mean's last digit. Outcomes whose rounding unit in $z$-units, $\kappa = 2a/SE_{\min}$ ($a$ the unit of the means), exceeds 0.5 are left out. Projection test: $J = 40$, $J_0 = 20$, centered at 1.96, 200 bootstrap repetitions; under $s_{\min}$ and $+$noise it rejects only when its breakdown statistic $\hat B$ exceeds the rounding allowance 0.13 times the mean of $\kappa^6$ over the reported outcomes (Propositions 3 and 4), and otherwise when $p \le .05$; the row ``$p \le .05$ only'' ignores the allowance.  EKW: the tests of \cite{Elliott}, as in the application, on the two-sided $p$-values below 0.15: LCM and Fisher, both of a concave $p$-curve CDF, the density discontinuity test at .05, the binomial test on $[.04,.05]$, and the Cox--Shi histogram tests of monotonicity (CS1) and 2-monotonicity with bounds (CS2B). $p$-curve: the test for left skew of Simonsohn, Nelson and Simmons (2014), by Stouffer's method. The two ``smallest $p$-value'' rows reject when the smallest of their tests' $p$-values is at most .05; they are not size-controlled tests. Every test treats trials as independent units. Bold: a size above .069, the upper end of the 95\% Monte Carlo band around 5\%. Data: \cite{Schwab2020}.\end{minipage}
\end{table}

\end{document}